\documentclass[aps,pra,amsmath,notitlepage,nofootinbib,twocolumn,superscriptaddress,noeprint]{revtex4-2}
\usepackage{array}
\usepackage{ragged2e}
\newcolumntype{M}[1]{>{\Centering\arraybackslash}p{#1}}
\usepackage{amsmath,amsfonts,amssymb,amsthm}
\usepackage{mathtools}
\usepackage[pdftex]{graphicx}
\usepackage[usenames, dvipsnames, svgnames, table]{xcolor}
\usepackage[colorlinks=true, citecolor=Blue, linkcolor=BrickRed, urlcolor=Brown]{hyperref}
\usepackage{txfonts}
\usepackage{mathrsfs}
\usepackage{bm}
\usepackage{multirow}
\usepackage{braket}
\usepackage{import}
\usepackage[ruled, noline, noend]{algorithm2e}

\usepackage{comment}

\newtheorem{theorem}{Theorem}
\newtheorem{lemma}{Lemma}
\newtheorem*{lemmaRestate}{Lemma}
\newtheorem*{theoremRestate}{Theorem}

\newtheorem{Note}{Note}

\newtheorem{corollary}{Corollary}
\theoremstyle{definition}
\newtheorem{definition}{Definition}
\newcommand{\be}{\begin{equation}}
\newcommand{\ee}{\end{equation}}
\newcommand{\ben}{\begin{eqnarray}}
\newcommand{\een}{\end{eqnarray}}
\newcommand{\bes}{\begin{subequations}}
\newcommand{\ees}{\end{subequations}}
\newcommand{\bF}{\begin{figure}}
\newcommand{\eF}{\end{figure}}

\newcommand{\orcid}[1]{\href{https://orcid.org/#1}{\includegraphics[height = 2ex]{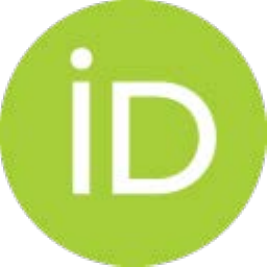}}}

\begin{document}

\title{Quantum Accreditation with Non-Clifford Two-qubit Gates}

\author{Andrew Jackson \orcid{0000-0002-5981-1604}}
\email{Andrew@ukjackson.com}
\affiliation{
School of Informatics, University of Edinburgh, Edinburgh, EH8 9AB, United Kingdom}
\affiliation{Department of Physics, University of Warwick, Coventry CV4 7AL, United Kingdom}
  
\author{Theodoros Kapourniotis \orcid{0000-0002-6885-5916}}
\thanks{The majority of this work was completed during the author’s affiliation with the Department of Physics, University of Warwick}
\email{Theodoros.Kapourniotis@stfc.ac.uk}
\affiliation{National Quantum Computing Centre, Rutherford Appleton Laboratory,
Harwell Campus,
Didcot,
OX11 0QX, United Kingdom}

\author{Animesh Datta \orcid{0000-0003-4021-4655}}
\email{Animesh.Datta@warwick.ac.uk}
\affiliation{Department of Physics, University of Warwick, Coventry CV4 7AL, United Kingdom}

%%%%%% Abstract %%%%%%

\begin{abstract}
	We develop a family of quantum accreditation protocols for quantum circuits with non-Clifford two-qubit gates.
	The latter includes families of gates such as the fSim and XY families of gates, native to existing hardwares.
    We provide practical and scalable protocols that upper-bound the total variation distance between the probability distributions of error-free and erroneous quantum computations.
    We also establish the robustness of our protocols to small perturbations and generalize Pauli twirling to non-Pauli single-qubit bases, which may be of independent interest.
\end{abstract}

\date{\today}
\maketitle

%%%%%% MAIN TEXT %%%%%%
\section{Introduction}
\label{IntroSec}
The respective limitations of classical and quantum computers are captured by distinct hierarchies of computational complexity. Thus, some problems that are intractable on classical computers are expected to be tractable on quantum computers~\cite{Harrow_2017}, most notably in cryptography~\cite{365700}
and simulating quantum systems in physics, chemistry, and material science~\cite{doi:10.1126/science.273.5278.1073, doi:10.1126/science.1113479, Bauer_2020, Chowdhury_2021, Jackson_2023}.
However, before they can achieve their full potential, quantum computers must pass through the NISQ (Noisy Intermediate-Scale Quantum) era~\cite{Preskill_2018}, when they will be small and subject to un-corrected~\cite{Devitt_2013} potentially computation-destroying noise. Even fault-tolerant quantum computers will not be noise- and error-free. Hence any fruitful utility of quantum computers relies on quantifying the impact of noise on quantum computers using real quantum hardware. This typically relies on measuring, managing, or mitigating the noise~\cite{Kandala2019, Ritter_2019, Ollitrault_2020}.  

One possible solution to this may be to characterize the hardware itself and the noise it experiences~\cite{lall2025reviewcollectionmetricsbenchmarks}.
Such characterizations are typically grouped in a family of protocols known as benchmarking~\cite{PRXQuantum.6.030202}. 
Benchmarking obtains a measure of how well a given hardware functions by running a set of circuits and inferring a performance metric \emph{for the hardware} from their outputs.
However, examining the performance of \emph{only} the hardware -- in general -- leaves open the question of determining if the \emph{specific} computation we are interested in is even close to correctly implemented unanswered.

Quantum accreditation answers to the above question. It is an experimental procedure to obtain an upper bound on the difference between the ideal (i.e., without error) and actual (i.e., with error possibly occurring) distributions generated by the measurements in a circuit execution. The circuits executed as part of an accreditation protocol always acts on the same number of qubits as the input circuit to be accredited, and has the same the depth. 

Quantum accreditation~\cite{Ferracin_2019} combines the best of quantum cryptographic protocols for verification~\cite{Gheorghiu2018} with the reality of quantum hardware. The former included concepts such as the quantum one-time-pad, and target and trap circuits. The latter is that single-qubit gates are typically the best component in quantum hardware. 
It is thus the first and only-known NISQ-practical verification protocol. It has been experimentally validated~\cite{Ferracin_2021, Liu2026}, showing that its assumptions held in practice.
Accreditation can be used to mitigate the error in computations~\cite{Mezher_2022}.
Accreditation was extended to work in tandem with error correction while introducing twirling for non-Clifford single-qubit gates~\cite{mills2025logicalaccreditationframeworkefficient}.
Finally, on the digital front, accreditation was returned to the adversarial model, which is very common in cryptographic verification protocols~\cite{jackson2025accreditationlimitedadversarialnoise,Ferracin_2019}.
Finally, quantum accreditation has also enabled the first verification of analogue quantum simulations~\cite{jackson2023accreditation, jackson2025improvedaccreditationanaloguequantum}.

Other protocols such as mirror circuit fidelity estimation~\cite{Proctor2022cep} claim to quantify the correctness of specific computations on quantum hardware. However, the fidelity of the circuit implementation does not quantify the correctness of the 
results obtained as does the variation distance used in quantum accreditation.

All existing accreditation protocols for digital/gate-based quantum computing require that all two-qubit gates be Clifford~\cite{Gottesman_1998}. Incidentally, mirror circuit fidelity estimation also requires the implementation of Clifford self-inverse two-qubit gates. Experimentally, of the numerous gate sets used in hardware~\cite{devereux2025quantumalgorithmssolvingdriftdiffusion}, most have two-qubit gates that are not Clifford. 
Instances include the fSim gate~\cite{Arute2019, PhysRevA.110.022608, Tsoukalas:2025ora, m5z3-jgk8, Cirq_Documentation_2025}, those based on specific parameter values like the $\sqrt{\text{iSWAP}}$ gate~\cite{Rasmussen_2020, Picard2024, m5z3-jgk8, q9sd-rfp6, 68m1-mjzy} or the Sycamore gate~\cite{Cirq_Documentation_2025a}, and others use controlled phase gates such as the Jaksch Gate~\cite{Jaksch_2000}.
Unfortunately, recompiling a quantum circuit using a gateset with Clifford two-qubit gates can be computationally expensive and increase the circuit depth factors of up to four. In the NISQ era, this can easily make computation inexecutable.

Twirling and accrediting circuits with two-qubit non-Clifford gates would thus be welcome for establishing the correctness and utility of NISQ computers. Additionally, the previously mentioned extensions of accreditation to the partial and fully fault-tolerant quantum computation, as well as the benchmarking of gatesets containing non-Clifford two-qubit gates~\cite{PhysRevA.92.060302, Cross2016}, may also benefit from twirling of and accrediting of circuits with two-qubit non-Clifford gates. 

Thus motivated, this paper develops quantum accreditation for two families of non-Clifford two-qubit gates: first, $\hat{\tau}$-decomposible gates (Def.~\ref{tauDecompDef}) and then XY-decomposible gates (Def.~\ref{def:XYinteractGate}). In each case twirling error occurring in such two-qubit gates is examined, and then the use of that family of two-qubit gate in accreditation -- similar to that in Ref.~\cite{Ferracin_2019} -- is demonstrated, using the immediately preceding twirling methods. The key results are presented in Table~\ref{TableOfResults}.
XY-decomposible gates are widely discussed and used~\cite{PhysRevA.67.032301, Abrams2020}. Although we could locate $\hat{\tau}$-decomposible gates being native to hardware, we note its potential use in tandem with probabalistic error correction~\cite{Chen:2025lgg}.

We begin by defining an expanded family of quantum accreditation protocols (QAPs). 
We define QAP as an algorithm applied on a triplet -- of 
a target computation, a trap computation, and an error model. Any QAP then provides an upper bound on the total variation distance between an ideal computation and its noisy, real-world implementation (Theorem~\ref{definitionWorksTheorem}).
In the course of presenting our results, we also generalize twirling CPTP error to stochastic error (in Appendix~\ref{TwirlingTauDeets}), following the approach of Ref.~\cite{Cai_2019}, so we can twirl error in any basis of a quantum device's Hilbert space or some subspaces thereof. 
We refer to this as generalized twirling and it comes at the cost -- in some cases -- of restricting the CPTP error that can be twirled.
\begin{table}
\centering
\includegraphics[width=\columnwidth]{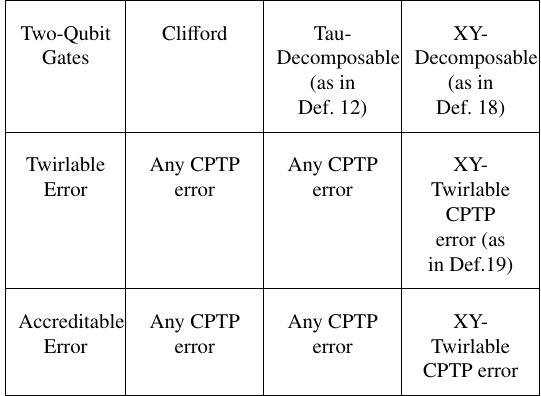}
\caption{Summary of our results on twirling and accrediting. All results assume the error model in Sec.~\ref{Sec:ErrorModel}. 
All CPTP maps act on both the system and its environment. By twirlable, we mean the error for which we present a method of effectively reducing it to stochastic error from a known finite set.
By accreditable, we mean error for which we present a protocol allowing for circuits experiencing that kind of error to be accredited, as in Def.~\ref{accreditationFormalDEf}. Formal formal definitions are given in Appendix~\ref{appendix:defsOFTwirlableAndAccreditable}. With an additional assumption (Sec.~\ref{subsec:ImprovedXYTwirling}), twirlable and accreditable errors of XY-Decomposable gates can both be upgraded to any CPTP error. }
\label{TableOfResults}
\end{table}

Additionally, we -- in Sec.~\ref{sec:WeaklyPerturb} -- present an enhancement and simplification of the argument for the robustness of accreditation protocols, first presented in Ref.~\cite{Ferracin_2021}. Now, any assumption a QAP depends on can be violated provided the resulting CPTP maps remain close (in terms of any sub-multiplicative norm bounded by one) to CPTP maps that do not violate the same assumption.
This is more formally stated in Corollary~\ref{neglectingDependenceCorrolary}.

\section{Formal Definition of Accreditation}
\label{FormDefSec}
\subsection{Basic Notation and Preliminary Definitions}
Let,
\begin{enumerate}
\item $\textit{Circ}$ be the set of all circuits.\\
\item $\big \{ \cdot \big \}$ denote a probability distribution such that its argument is always the value of any random variable defined by this distribution, i.e., a sample from the distribution $\{ \vert 0 \rangle \}$ will always return $\vert 0 \rangle$\\
\item $\mathrm{P}_{\mathrm{oly}}$ be the set of all polynomials\\
\item  $\big \vert \cdot \big \vert$ denote the number of qubits in the argument circuit\\
\item $\big \vert \big \vert \cdot \big \vert \big \vert$  denote the number of gates in the argument circuit
\end{enumerate}
To formally define accreditation, we must first formally define the concepts that accreditation addresses.
\begin{definition}
    An \underline{execution} of a given circuit is any -- potentially noisy, as defined in Def.~\ref{def:NoiseandError} -- physical process intended to obtain a sample from the probability distribution generated by the measurements in that circuit. For any circuit, $\mathcal{C}$, any execution of it will be denoted as $\Tilde{\mathcal{C}}$; or with indices if there are multiple distinct executions (which may vary in implementation -- and hence error -- but share the same ideal\footnote{Meaning that no error occurs.} outputs).
\end{definition}

\begin{definition}
    \label{def:NoiseandError}
    \underline{Noise} is any process that applies CPTP maps (that we refer to as \underline{error}) alongside the intended operations during the execution of a circuit. 
    
    For example, if $\hat{A}$ is a gate within a circuit, the noise may add an Pauli X gate immediately after it, as shown below:
    \begin{align}
        \hat{A} \longrightarrow \hat{X} \hat{A}.
    \end{align}
    The noise adds CPTP maps after every gate in the circuit and state preparation, but immediately before measurement.
\end{definition}

Technically an execution must be defined relative to an error model that defines the noise -- and hence error -- occurring in that execution. Herein, error models are kept general and abstract until Sec.~\ref{Sec:ErrorModel} and thereafter all executions are assumed to be afflicted by noise conforming to the error model presented in Sec.~\ref{Sec:ErrorModel} to describe the noise they experience.
Due to noise changing the distributions circuit executions sample from, it is important to have some measure of how much two given distributions differ. The metric we will use to quantify this is the total variation distance, as defined in Def.~\ref{def:total-variation}, and the key application of it for this paper is in the ideal-actual variation distance, as defined in Def.~\ref{def:ideal-Actual}.
\begin{definition}
\label{def:total-variation}
For any two probability measures, $P$ and $\Tilde{P}$, over a sample space, $\Omega$, the \underline{total variation distance} between $P$ and $\Tilde{P}$ is:
    \be
\label{eq:tvd}
      \mathrm{VD} \left( P, \Tilde{P} \right)
      \equiv
        \dfrac{1}{2} \sum_{s \in \Omega} \bigg \vert P \left( s \right) - \Tilde{P}\left( s \right) \bigg \vert.
       \ee
\end{definition}

\begin{definition}
    \label{def:IdealDistribution}
    For any circuit, $\mathcal{C}$, define the distribution that an error-free execution of it samples from by \underline{$\mathcal{O}_I \left [ \mathcal{C} \right ]$.}
\end{definition}

\begin{definition}
\label{def:ideal-Actual}
    For any circuit, $\mathcal{C}$, and an execution of it, $\Tilde{\mathcal{C}}$, the \underline{ideal-actual variation distance} of that execution (denoted as $\nu \left [ \Tilde{\mathcal{C}} \right ]$) is the total variation distance between the probability distribution the execution would sample from if there were no error (i.e., the ideal case), $\mathcal{O}_I \left [ \mathcal{C} \right ]$, and the distribution it actually samples from (which may be different, due to noise).
\end{definition}

We now switch focus slightly to consider various sets of operators/gates that we will use. 
This begins with Def.~\ref{def:PauliGates}.
\begin{definition}
    \label{def:PauliGates}
    \underline{$\mathbb{P}_1$} $\subseteq \mathbb{U}(2)$ is the set of all single-qubit Pauli gates and the identity. \underline{$\mathbb{P}^{\otimes N}_1$} $\subseteq \mathbb{U}(2^N)$ denotes the set of all $N$-fold strings of -- not necessarily identical -- elements from $\mathbb{P}_1$.
\end{definition}
Another set of gates that are very useful are the Clifford gates. Clifford gates were originally defined differently, in Ref.~\cite{Gottesman_1998}, but, for our purposes, the following is more convenient (and is equivalent to the original definition):
\begin{definition}
    A two-qubit gate, $\mathcal{G}$, is \underline{Clifford} if and only if $\forall \hat{p} \in \mathbb{P}_1^{\otimes 2}$, $\hat{\mathcal{G}}^{\dagger} \hat{p} \hat{\mathcal{G}} \in \mathbb{P}_1^{\otimes 2}$, up to an overall phase.
\end{definition}
To ease discussing the relationship between the Pauli gates and Clifford gates, we formalize it in Def.~\ref{def:pushing}.
\begin{definition}
    \label{def:pushing}
    A gate, $\hat{g}$, is said to be \underline{pushed} through another gate, $\hat{G}$, in a circuit if an occurrence of $\hat{G} \hat{g}$ in that circuit is replaced with $\hat{g}' \hat{G}$, where $\hat{G} \hat{g}$ and $\hat{g}' \hat{G}$ are equivalent and $\hat{g}'$ is another gate known as \underline{the result of pushing $\hat{g}$ through $\hat{G}$}, i.e.,
    \begin{align}
        g' = Gg G^{\dagger} \Rightarrow \hat{g}' \hat{G} = \hat{G} \hat{g}.
    \end{align}
    Note that this generalizes to gates being pushed through each other in either direction.
\end{definition}
The final definition, Def.~\ref{def:minimumLooseness}, required before presenting our formal QAPs is more general than the previous ones. It is a consequence of using sampling -- which inevitably comes with some, bounded, sampling error -- to estimate an upper bound on some quantity.
\begin{definition}
\label{def:minimumLooseness}
    If $x \in \mathbb{R}$ is the true value of some unknown quantity to be upper bounded and $\Tilde{x}$ is the minimum value that a given stochastic process, guaranteed to upper bound $x$, returns with non-zero probability. Then the \underline{minimum looseness} of that process when bounding $x$, $\theta_0$, is defined by:
    \begin{align}
      \big \vert  x - \Tilde{x} \big \vert > \theta_0.
    \end{align}
    This value limits how tight the bound on $x$, obtained by a specific stochastic process, can be with non-zero probability. 
  \end{definition}

\subsection{Quantum accreditation protocol (QAP)}
\label{FormalDefSubSec}

The aim of any QAP is to provide a bound on the total variation distance between the probability distribution the execution of a circuit provides samples from and the distribution it is intended to provide samples from. This can alternatively be expressed as bounding the total variation distance between the ideal and erroneous (averaged over all errors that occur) probability distributions.

We first note that any QAP is defined and proven to function correctly relative to a specific error model.
 We formally define QAPs in Def.~\ref{accreditationFormalDEf}
via a series of conditions on a set of circuits and the algorithms that generate them. This will hopefully enable -- in the future -- more general notions of accreditation. In Theorem~\ref{definitionWorksTheorem}, we show that these conditions imply that a standard protocol (Protocol~\ref{StandardAccAlg}), if given algorithms the conditions -- in Def.~\ref{accreditationFormalDEf} -- apply to, functions as a QAP.

\begin{definition}
\label{accreditationFormalDEf} 
\begin{widetext}
A \textit{quantum accreditation protocol (QAP)} for a triplet $\{\mathcal{E}, P_{\mathrm{targ}}, P_{\mathrm{trap}}\}$ consisting of an
 error model $\mathcal{E}$ and two polynomial time classical algorithms $P_{\mathrm{targ}}$ and $P_{\mathrm{trap}}$, is an algorithm for using them (given as Algorithm~\ref{StandardAccAlg}), such that: 
\begin{enumerate}
\item[ $\underline{\mathbf{1}}:$]
$\forall \mathcal{C} \in \textit{Circ}$, $\exists \mathbb{C}_{\mathrm{targ}} \subseteq \textit{Circ}$, such that: $\forall \mathcal{C}_{\mathrm{targ}} \in \mathbb{C}_{\mathrm{targ}}$,
\begin{enumerate}
\item[  $\mathbf{a)}$]
$P_{\mathrm{targ}} \left( \mathcal{C} \right)$ samples from $ \mathbb{C}_{\mathrm{targ}}$ according to distribution $\mu_{\text{targ}}$
\item[ $\mathbf{b)}$]
$\exists f \in \mathrm{P}_{\mathrm{oly}}$, such that $\left \vert \mathcal{C}_{\mathrm{targ}} \right \vert \leq f \left ( \left\vert \mathcal{C} \right \vert \right)$
\item[  $\mathbf{c)}$]
$\exists f_t \in \mathrm{P}_{\mathrm{oly}}$, such that $\big \vert \big \vert \mathcal{C}_{\mathrm{targ}} \big \vert \big \vert \leq f_t \left ( \left \vert \left\vert \mathcal{C} \right \vert \right \vert \right)$
\item[ $\mathbf{d)}$]
$\mathcal{O}_I \left [ \mathcal{C} \right] = \mathcal{O}_I \left[ \mathcal{C}_{\mathrm{targ}} \right]$
\item[$\mathbf{e)}$] Executing many circuits -- experiencing error conforming to $\mathcal{E}$ -- sampled from $\mathbb{C}_{\text{targ}}$ according to distribution $\mu_{\text{targ}}$ is equivalent, in terms of probability distribution of the measurement outcomes, to executing a circuit in $\mathbb{C}_{\text{targ}}$ with all error being stochastic error from a known set.
\end{enumerate}
% Condition 2
\item[ $\underline{\mathbf{2}}:$]
$\forall \mathcal{C} \in \textit{Circ}$, $\exists \mathbb{C}_{\mathrm{trap}} \subseteq \textit{Circ}$, such that:
\begin{enumerate}
\item[ $\mathbf{a)}$]
$P_{\mathrm{trap}} \left( \mathcal{C} \right)$ samples $\mathbb{C}_{\mathrm{trap}}$ according to distribution, $\mu_{\text{trap}}$
\item[ $\mathbf{b)}$]
$\exists f^{\prime} \in \mathrm{P}_{\mathrm{oly}}$, such that, $\forall \mathcal{C}_{\mathrm{trap}} \in \mathbb{C}_{\mathrm{trap}}$, $\left \vert \mathcal{C}_{\mathrm{trap}} \right \vert \leq f^{\prime} \left ( \left \vert \mathcal{C} \right \vert \right)$
\item[ $\mathbf{c)}$]
$\exists f^{\prime}_t \in \mathrm{P}_{\mathrm{oly}}$, such that $\left \vert \left \vert \mathcal{C}_{\mathrm{trap}} \right \vert \right \vert \leq f^{\prime}_t \left ( \left \vert \left \vert \mathcal{C} \right \vert \right \vert \right)$
\item[ $\mathbf{d)}$]
$\mathcal{O}_I \left[ \mathcal{C}_{\mathrm{trap}} \right] = \left \{ m \right\}$ is efficiently computable  classically, where $m$ is some output of the measurements in the traps
\item[ $\mathbf{e)}$] Executing many circuits -- each experiencing error conforming to $\mathcal{E}$ -- sampled from $\mathbb{C}_{\text{trap}}$ according to distribution $\mu_{\text{trap}}$ is equivalent, in terms of probability distribution of the measurement outcomes, to executing a circuit in $\mathbb{C}_{\text{trap}}$ with all error being stochastic error from a known set
\item[ $\mathbf{f)}$] Any error occurring in a trap chosen from $\mathbb{C}_{\mathrm{trap}}$ according 
 to distribution $\mu_{\text{trap}}$ causes the trap to not return $m$, with probability at least $k \in (0,1]$ (where $k$ is known)
\end{enumerate}

\item[$\underline{\mathbf{3}}:$] Trap and target circuits are constructed such that they only differ in aspects that the error model, $\mathcal{E}$, asserts do not affect any errors in an execution. So the error is identical in all elements of $\mathbb{C}_{\text{targ}} \cup \mathbb{C}_{\text{trap}}$
\end{enumerate}
\end{widetext}

Key parts of Def.~\ref{accreditationFormalDEf} are conditions $\mathbf{1e}$ and $\mathbf{2e}$. These mean errors in the trap and target circuit executions can be considered stochastic and allow the concept in Def.~\ref{def:probOfError} be be meaningful for these circuit executions.
Note that $k$ and $m$ are features of $P_{\mathrm{targ}}$ and $P_{\mathrm{trap}}$, while $\theta \in [0,1]$ and $\alpha \in [0,1]$ are inputs to Algorithm~\ref{StandardAccAlg}.

\begin{figure}
    \centering
\begin{algorithm}[H]
%\SetAlgoLined
$\mathbf{Inputs:}$ \\
$\bullet$ A circuit, $\mathcal{C}$, meeting any required conditions of the traps and targets.\\
$\bullet$ The relevant $P_{\mathrm{targ}}$, $P_{\mathrm{trap}}$, $m$, and $k$ from Def.~\ref{accreditationFormalDEf}.\\
$\bullet$ The minimum looseness, $\theta / k$, of the ideal-actual variation distance to return (where $k$ is as in $\textbf{1f}$, above, and is known).\\
$\bullet$ A required confidence in the above minimum looseness, $\alpha$.\\
 %\hline 
 \begin{enumerate} 

    \item Calculate $N_l = \left \lceil \dfrac{2}{\theta^2} \ln{\left( \dfrac{2}{1 - \alpha} \right)} \right \rceil + 1$

    \item Use $P_{\mathrm{trap}}$ on input $\mathcal{C}$ to generate a set of traps, $N_l$ = $\left \{ \mathcal{C}_j \right\}_{j = 1}^{N_l} \subseteq \mathbb{C}_{\mathrm{trap}}$.

    \item Use $P_{\mathrm{targ}}$ on input $\mathcal{C}$ to generate a target circuit, $\mathcal{C}^{\prime} \in \mathbb{C}_{\mathrm{targ}}$

    \item Execute every circuit in $\left \{ \mathcal{C}_j \right \}_{j = 1}^{N_l} \cup \left \{ \mathcal{C}^{\prime} \right \}$ in a random order and record the results.

    \item $\mathrm{TargResult}$ = output of executing $\mathcal{C}^{\prime}$.
    \item $\mathrm{TrapResult}$ = fraction of $\left \{ \Tilde{\mathcal{C}}_j \right \}_{j = 1}^{N_l}$ that does not output $m$.
\end{enumerate}
%\hline
$\mathbf{Return}:$ $\mathrm{TargResult}$, $\left( \mathrm{TrapResult} + \theta \right)/k.$
\caption{Standard Accreditation Algorithm
 \label{StandardAccAlg}}
\end{algorithm}
\end{figure}
\end{definition}

\begin{definition}
    \label{def:probOfError}
    For any execution of any circuit afflicted by stochastic error, the \underline{probability of error} is the probability that \emph{any} of the stochastic channels modeling the error apply an erroneous gate.
\end{definition}

\subsubsection{Proof of Correctness of Algorithm~\ref{StandardAccAlg}}

We now prove that Def.~\ref{accreditationFormalDEf} and Algorithm~\ref{StandardAccAlg} achieve the goals of accreditation. 
\begin{theorem}
\label{definitionWorksTheorem}
    Given algorithms $P_{\mathrm{targ}}$ and $P_{\mathrm{trap}}$ that generate traps and targets, respectively, as defined in Def.~\ref{accreditationFormalDEf}, and a circuit, $\mathcal{C}$; assuming all circuit executions experience error conforming to the error model that the generated traps and target circuits are designed to function under; Algorithm~\ref{StandardAccAlg} provides a sample from a probability distribution and a bound on the total variation distance between that distribution and the distribution generated by the measurements in $\mathcal{C}$ if no error occurs in its execution. 
\end{theorem}
\begin{proof}
    For a given circuit, $\mathcal{C}$, Algorithm~\ref{StandardAccAlg} efficiently obtains a sample from the distribution generated by the measurements of an implementation of a circuit in $\mathbb{C}_{\mathrm{targ}}$ which is: efficiently obtainable (using $P_{\mathrm{targ}}$), efficiently executable, and generates the same distribution with its measurements -- when no error occurs -- as $\mathcal{C}$; due to Conditions $\mathbf{1a-c}$.
    
    Implementing $\Tilde{\mathcal{C}}_{\mathrm{targ}}$, when $\mathcal{C}_{\mathrm{targ}} \in \mathbb{C}_{\mathrm{targ}}$, is equivalent, in the ideal case, to implementing $\Tilde{\mathcal{C}}$, due to Condition $\mathbf{1d}$. Hence the returned samples, in the ideal case, are from the intended distribution.
    
We now prove that the bound Algorithm~\ref{StandardAccAlg} returns is truly a bound on the ideal-actual variation distance of the target.
    
    Conditions $\mathbf{2a-c}$ imply trap circuits are BQP-computations if the input circuit is a BQP-computation and can be efficiently obtained via $P_{\mathrm{targ}}$. Condition $\mathbf{2d}$ means if no error occurs in a particular trap, we know -- with certainty. Likewise, Condition $\mathbf{2f}$ means, if error does occur, the trap gives an incorrect output with probability at least $k$.

    Therefore, the probability of traps returning incorrect outcomes can be estimated with minimum looseness $\theta$, (as in Def.~\ref{def:minimumLooseness}) with confidence $\alpha \in [0,1]$ using $ N_l = \bigg \lceil \dfrac{2}{\theta^2} \ln{\left( \dfrac{2}{1 - \alpha} \right)} \bigg \rceil + 1$ traps, due to Hoeffding's inequality~\cite{doi:10.1080/01621459.1963.10500830}.

    Due to Condition $\mathbf{2f}$ all error in a trap is effectively stochastic error. Then, as Condition $\mathbf{1e}$ keeps the target and trap circuits experiencing identical error despite this and Condition $\mathbf{3}$ implies each trap execution has the same probability of error; the probability of error which is meaningful due to $\mathbf{2e}$) occurring in a trap can be upper bounded -- with confidence $\alpha$ -- as
    \begin{align}
         \mathbb{P} \left( \text{Error in trap}  \right) \leq \frac{1}{k} \left( \mathbb{P} \left(\mathrm{TrapResult}  \right)  + \theta \right),
    \end{align}
    where $\mathrm{TrapResult}$ = fraction of $\{ \Tilde{\mathcal{C}}_j \}_{j = 1}^{N_l}$ that does not output $m$ As in Algorithm~\ref{StandardAccAlg}.
    
    The penultimate step in this proof follows from Condition $\mathbf{3}$. As all error in the target and traps is identical, the probability of error in the traps is the probability of error in the target. Therefore:
    \begin{align}
        \label{probTargBound}
        \mathbb{P} \left( \text{Error in target} \right)
        \leq
        \frac{1}{k} \left( \mathbb{P} \left(\mathrm{TrapResult}  \right)   + \theta \right),
    \end{align}
    with confidence $\alpha \in [0,1]$. Similarly, Eqn.~\eqref{probTargBound2} is meaningful due to $\mathbf{1e}$. Finally, Lemma~\ref{stochAndVDLemma} implies that, using Eqn.~\eqref{probTargBound},
    \begin{align}
        \label{probTargBound2}
        \nu \left [ \Tilde{\mathcal{C}}_{\mathrm{targ}} \right ] 
         \leq
        \frac{1}{k} \left( \mathbb{P} \left(\mathrm{TrapResult}  \right)   + \theta \right),
    \end{align}
    again, with confidence $\alpha$.
\end{proof}
The derivation of Eqn.~\eqref{probTargBound2} uses Lemma~\ref{stochAndVDLemma}, the proof of which is in Ref.~\cite[Appendix~1]{Ferracin_2021}.
\begin{lemma}
    \label{stochAndVDLemma}
    For any execution of a circuit, if all error occurring is stochastic error, then the ideal-actual variation distance of the execution of the circuit is upper bounded by the probability of error occurring during the execution (regardless of if it is detected).
\end{lemma}

\section{Error model}
\label{Sec:ErrorModel}
Unless and until stated otherwise, we will use a consistent error model throughout this paper. This model, also used in Ref.~\cite{Ferracin_2021} is as follows:
\begin{enumerate} 
\item[$\mathbf{N1}$:] Noise in state preparation, two-qubit gates, and measurement is a Completely Positive
Trace Preserving (CPTP) map acting on the system and its environment, acting immediately before or after the ideal implementation of the state preparation, gate, or measurement the error occurs in.
\item[$\mathbf{N2}$:] Noise in any gate is independent of which single-qubit gates are in the circuit\footnote{This assumption is motivated by the empirical reality of existing quantum hardware. The noise in any part of the execution of a circuit is independent of which single qubit gates are executed at any point in the circuit (e.g., if the gate in a given location is a Pauli $\hat{X}$ or Hadamard gate). The noise (at any point of the circuit) may still depend on \emph{if} there is a single-qubit gate and which qubits it acts on, but not \emph{which} single-qubit gate it is.}. The error can still depend on the location, within the circuit, of single-qubit gates.
\end{enumerate}

Despite its brevity, we present this as an independent section for ease of future reference and clarity.
It is also a stand against the frequent practice in the literature of not displaying the assumptions prominently.

\section{Twirling and Accrediting tau-Decomposible Gates}
\label{TauDecSec}
 We now move away from our generalised QAP and towards specifics: namely, we aim to accredit the executions of circuits that feature non-Clifford two-qubit gates. We start by examining how we can twirl error occurring in these gates to stochastic error and then move on to accrediting the execution of circuits that contain these gates.
\subsection{tau-Decomposible Gate Definition and Properties}
The first set of non-Clifford two-qubit gates we consider is the set of $\hat{\tau}$-decomposible gates, as defined below.
\begin{definition}
    \label{tauDecompDef}
    A two-qubit gate, $\hat{\mathcal{G}} \in \mathbb{U}(4)$, is \underline{$\hat{\tau}$-decomposible} if and only if it can be expressed as:
    \begin{align}
    \label{GDef}
    \hat{\mathcal{G}} = \left( \hat{\tau}_1 \otimes \hat{\tau}_2 \right)^{\dagger} \circ \hat{\mathcal{M}} \circ \left( \hat{\tau}_1 \otimes \hat{\tau}_2 \right),
\end{align}
where $\hat{\mathcal{M}}$ is a two-qubit Clifford gate that leaves $\vert 0 \rangle^{\otimes 2}$ and $\hat{H}^{\otimes 2} \vert 0 \rangle^{\otimes 2}$ unchanged -- up to an overall phase -- and $\hat{\tau}_1$, $\hat{\tau}_2$ $\in \mathbb{U} (2)$. Additionally, $\circ$ represents the composition of operators but is often omitted herein for conciseness.
\end{definition}
We call Eqn.~\eqref{GDef} the $\hat{\tau}$-decomposition of $\mathcal{G}$. It does not necessarily exist for all two-qubit gates. Def.~\ref{tauDecompDef} additionally inspires a further definition.
\begin{definition}
    \label{PauliVesicleSetDef}
    $\forall \hat{\tau} \in \mathbb{U}(2)$, the corresponding \underline{$\hat{\tau}$-Pauli vesicle set}, $\Gamma_{\hat{\tau}} \in \mathbb{U}(2)$, is defined as:
    \begin{align}
        \Gamma_{\hat{\tau}}
        &=
        \big \{ \hat{\tau}^{\dagger} \circ \hat{P}_j \circ \hat{\tau} \text{ } \vert \text{ } \hat{P}_j \in \mathbb{P}_1 \big \}.
    \end{align}
\end{definition}
We can additionally consider a tensor product of two $\hat{\tau}$-Pauli vesicle sets, where an element from one $\hat{\tau}$-Pauli vesicle set acts on one qubit and an element from another acts on a second qubit. It is of further use if we specify that the two $\hat{\tau}$-Pauli vesicle sets are those that correspond to the single-qubit gates, $\hat{\tau}_1$ and $\hat{\tau}_2$, in the $\hat{\tau}$-decomposition of a specific two-qubit gate. 
For this purpose, we next define $\Gamma_{\hat{\mathcal{G}}}$.
\begin{definition}
    \label{def:GammaG}
    Define \underline{$\Gamma_{\hat{\mathcal{G}}}$} $\subseteq \mathbb{U}(2)$, using the same $\hat{\tau}_1$ and $\hat{\tau}_2$ from the decomposition in Eqn.~\eqref{GDef}, as:
    \begin{align}
        \Gamma_{\hat{\mathcal{G}}}
        &=
        \Gamma_{\hat{\tau}_1} \otimes \Gamma_{\hat{\tau}_2},
    \end{align}
    where, for any sets of operators, $A$ and $B$, $A \otimes B = \{ \hat{a} \otimes \hat{b} \text{ } \vert \text{ } \hat{a} \in A, \hat{b} \in B \}$.
\end{definition}
A key property of a $\hat{\tau}$-decomposible gate, $\hat{\mathcal{G}} \in \mathbb{U}(2)$, and the corresponding $\Gamma_{\hat{\mathcal{G}}}$, is that pushing elements of $\Gamma_{\hat{\mathcal{G}}}$ through $\hat{\mathcal{G}}$ maps those elements of $\Gamma_{\hat{\mathcal{G}}}$ to elements of $\Gamma_{\hat{\mathcal{G}}}$.
This property is proven in Lemma~\ref{GgateCommuteLemma} which first requires Def.~\ref{def:postCommutePaulis}. 

Viewing Def.~\ref{def:postCommutePaulis} informally, $\Bar{P}^{(1)}_{j, \hat{\mathcal{M}}}$ can be viewed as the result of pushing (as in Def.~\ref{def:pushing}) $\hat{P}_j \otimes \hat{I}$ through $\hat{\mathcal{M}}$, and $\Bar{P}^{(2)}_{j, \hat{\mathcal{M}}}$ is the result of pushing $\hat{I} \otimes \hat{P}_j$ through $\hat{\mathcal{M}}$.
We now state Lemma~\ref{GgateCommuteLemma} (which is proven in Appendix~\ref{appendixProofofLemmaGgateCommuteLemma}), as promised.

\begin{figure*}[t]
    \includegraphics[]{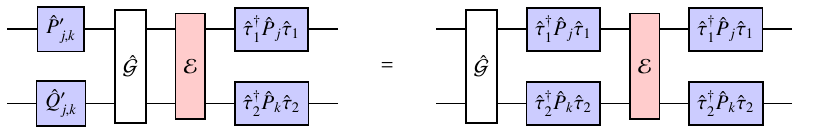}
\caption{Depiction of how we use $\hat{\mathcal{G}}$ normalizing $\Gamma_{\hat{\mathcal{G}}}$ to twirl CPTP error $\mathcal{E}$ as in Eqn.~\eqref{Eqn:GoverningForPQ'}). Note that we use red gates to denote CPTP maps and blue gates to denote gates that are chosen randomly from a specific set. In this case, that set is $\Gamma_{\hat{\mathcal{G}}}$ (drawn uniformly) for the gates on the right-hand of their respective sub-circuits, with the $P'_{j,k}$ and $Q'_{j,k}$ gates then selected to match them as in Eqn.~\eqref{Eqn:GoverningForPQ'}. We abuse notation to display the CPTP maps as gates (depicting error) in a circuit. As mitigation, we highlight the CPTP maps in red, denoting that they are not unitaries.}
\label{surroundError}
\end{figure*}

\begin{lemma}
\label{GgateCommuteLemma}
For any $\left(  \hat{\tau}_1^{\dagger} \hat{P}_j \hat{\tau}_1 \otimes \hat{\tau}_2^{\dagger} \hat{P}_k \hat{\tau}_2 \right) \in \Gamma_{\hat{\mathcal{G}}}$, there exists a $\hat{P}^{\prime}_{j,k} \in \Gamma_{\hat{\tau}_1}$ and $\hat{Q}^{\prime}_{j,k} \in \Gamma_{\hat{\tau}_2}$, such that:
\begin{align}
     \left(  \hat{\tau}_1^{\dagger} \hat{P}_j \hat{\tau}_1 \otimes \hat{\tau}_2^{\dagger} \hat{P}_k \hat{\tau}_2 \right) \circ \hat{\mathcal{G}}
     &=
     \label{PprimeQprimeMainProperty}
    \hat{\mathcal{G}} \circ \left( \hat{P}^{\prime}_{j,k} \otimes \hat{Q}^{\prime}_{j,k} \right).
\end{align}
\end{lemma}
Exact formulas for $\hat{P}^{\prime}_{j,k}$ and $\hat{Q}^{\prime}_{j,k}$ can be found in  Eqn.~\eqref{PPrimeDef} and Eqn.~\eqref{QPrimeDef}, respectively.

\subsection{Twirling Error in tau-Decomposible Gates}

The above properties of easily-computable $\hat{P}^{\prime}_{j,k}, \hat{Q}^{\prime}_{j,k} \in \mathbb{U}(2)$ become useful when $\hat{\mathcal{G}} \in \mathbb{U}(4)$ is implemented incorrectly, resulting in an erroneous implementation of $\hat{\mathcal{G}}$ that we denote as $\Tilde{G}$. In accordance with the assumption $\mathbf{N1}$, we model the erroneous implementation as the correct gate, $\hat{\mathcal{G}}$, followed by an arbitrary CPTP map~\cite{Keyl_2002} that we denote $\mathcal{E}$ i.e. $\Tilde{G} = \mathcal{E} \circ \hat{\mathcal{G}}$. We note that it is an abuse of notation but it is done for conciseness and clarity: we read $\mathcal{E} \circ \hat{\mathcal{G}}$ as the unitary $\hat{\mathcal{G}}$ being applied followed by the CPTP map $\mathcal{E}$.

Surrounding the erroneously implemented gate, $\Tilde{G}$, with $\left( \hat{\tau}_1^{\dagger} \hat{P}_j \hat{\tau}_1 \otimes \hat{\tau}_2^{\dagger} \hat{P}_k \hat{\tau}_2 \right)$ and the corresponding $\left( \hat{P}^{\prime}_{j,k} \otimes \hat{Q}^{\prime}_{j,k} \right)$, as in Eqn.~\eqref{Eqn:GoverningForPQ'} is equivalent to conjugating the CPTP error of $\Tilde{G} = \mathcal{E} \circ \hat{\mathcal{G}}$ by $\left(  \hat{\tau}_1^{\dagger} \hat{P}_j \hat{\tau}_1 \otimes \hat{\tau}_2^{\dagger} \hat{P}_k \hat{\tau}_2 \right)$. This is as, using Lemma~\ref{GgateCommuteLemma}:
\begin{align}
    \label{Eqn:GoverningForPQ'}
    &\left(  \hat{\tau}_1^{\dagger} \hat{P}_j \hat{\tau}_1 \otimes \hat{\tau}_2^{\dagger} \hat{P}_k \hat{\tau}_2 \right) \circ \Tilde{G} \circ \left( \hat{P}^{\prime}_{j,k} \otimes \hat{Q}^{\prime}_{j,k} \right)\\
     &=
     \left(  \hat{\tau}_1^{\dagger} \hat{P}_j \hat{\tau}_1 \otimes \hat{\tau}_2^{\dagger} \hat{P}_k \hat{\tau}_2 \right) \circ \mathcal{E} \circ \left(  \hat{\tau}_1^{\dagger} \hat{P}_j \hat{\tau}_1 \otimes \hat{\tau}_2^{\dagger} \hat{P}_k \hat{\tau}_2 \right) \circ \hat{\mathcal{G}}. \nonumber
\end{align}
Eqn.~\eqref{Eqn:GoverningForPQ'} is useful when we attempt to twirl the error as now we can place arbitrary elements of $\Gamma_{\hat{\mathcal{G}}}$ around \emph{just} the CPTP error occurring in an erroneous implementation of $\hat{\mathcal{G}}$ (i.e., around just the error and not the gate itself, as depicted in Fig.~\ref{surroundError}).
Sec.~\ref{TauDecSec} thus far may be summarized pictorially as in Fig.~\ref{surroundError}.

Practically, the most important feature of Fig.~\ref{surroundError} is that $\hat{P}^{\prime}_{j,k}$ and $\hat{Q}^{\prime}_{j,k}$ are single-qubit gates, and are subject to the error rates typical of single-qubit gates and any assumptions made about single-qubit gates e.g., $\mathbf{N2}$ in Sec.~\ref{Sec:ErrorModel}.

By choosing values of $j$ and $k$ uniformly at random to construct $ \left( \hat{\tau}_1^{\dagger} \hat{P}_j \hat{\tau}_1 \right) \otimes \left( \hat{\tau}_2^{\dagger} \hat{P}_k \hat{\tau}_2 \right) \in \Gamma_{\hat{\mathcal{G}}}$ to surround the error in the erroneous implementation, $\Tilde{G}$, of $\hat{\mathcal{G}}$ (as in Fig.~\ref{surroundError}), that surrounded CPTP error -- in the erroneous implementation of $\hat{\mathcal{G}}$ -- is effectively reduced to stochastic error comprising of elements from $\Gamma_{\hat{\tau}_1}$ and $\Gamma_{\hat{\tau}_2}$.

As $\Gamma_{\hat{\mathcal{G}}} = \Gamma_{\hat{\tau}_1} \otimes \Gamma_{\hat{\tau}_2}$ we term this error $\Gamma_{\hat{\mathcal{G}}}$-stochastic error. This process of reducing general CPTP error to $\Gamma_{\hat{\mathcal{G}}}$-stochastic error is expressed more formally in Lemma \ref{twirlLemma}, which is proven in Appendix \ref{TwirlingTauDeets}.
\begin{lemma}
    \label{twirlLemma}
    If any CPTP map is conjugated by uniformly random elements from $\Gamma_{\hat{\mathcal{G}}}$, it is effectively reduced to $\Gamma_{\hat{\mathcal{G}}}$-stochastic error.
\end{lemma}
This means that -- in the aggregate -- the distribution the circuit provides a sample from is transformed to one provided by the same \emph{intended} circuit but experiencing only stochastic Pauli error.
We can now consolidate our considerations thus far -- in Sec.~\ref{TauDecSec} -- into a single definition (Def.~\ref{twirlDef}) that encapsulates our generalization of Pauli twirling.
\begin{definition}{$\left( \hat{\tau}_1, \hat{\tau}_2 \right)$-twirling a two-qubit gate}\( \\ \)
    \label{twirlDef}
    An instance of a two-qubit gate, $\hat{\mathcal{G}}$, with a $\hat{\tau}$-decomposition, is \underline{$\left( \hat{\tau}_1, \hat{\tau}_2 \right)$-twirled} by uniformly choosing $\hat{\tau}_1^{\dagger} \hat{P}_j \hat{\tau}_1 \otimes \hat{\tau}_2^{\dagger} \hat{P}_k \hat{\tau}_2 \in \Gamma_{\hat{\mathcal{G}}}$ and replacing $\hat{\mathcal{G}}$ with:
    \begin{align}
          \left( \hat{\tau}_1^{\dagger} \hat{P}_j \hat{\tau}_1 \otimes \hat{\tau}_2^{\dagger} \hat{P}_k \hat{\tau}_2 \right)\circ \hat{\mathcal{G}} \circ
          \left( \hat{P}^{\prime}_{j,k} \otimes \hat{Q}^{\prime}_{j,k} \right),
    \end{align}
    where $\hat{P}^{\prime}_{j,k} $ and $\hat{Q}^{\prime}_{j,k}$ are as implied by which $\hat{\tau}_1$ and $\hat{\tau}_2$ gates are used.
\end{definition}
The net effect of $\left( \hat{\tau}_1, \hat{\tau}_2 \right)$-twirling every instance of a $\hat{\tau}$-decomposible gate, $\hat{\mathcal{G}}$, occurring within a circuit, $\mathcal{C}$, is that all error (originating in a $\hat{\mathcal{G}}$ gate) in $\mathcal{C}$ is effectively reduced to $\Gamma_{\hat{\mathcal{G}}}$-stochastic error.
$(I,I)$-twirling reduces to the conventional Pauli twirling used in Ref~[13.14] for accrediting circuits with only Clifford two-qubit gates.

Combining  Lemma~\ref{twirlLemma} with the previously established ability of $\hat{\mathcal{G}}$ to normalize its respective $\Gamma_{\hat{\mathcal{G}}}$ results in our first main result on twirling.
\begin{theorem}
    \label{TwirlingTheorem}
    Assuming $\mathbf{N1}$ and $\mathbf{N2}$, any two-qubit gate, $\hat{\mathcal{G}} \in \mathbb{U}(4)$, with a $\hat{\tau}$-decomposition can have its error effectively reduced to $\Gamma_{\hat{\mathcal{G}}}$-stochastic error with the addition of only single-qubit gates around it.
\end{theorem}
The solitary outstanding issue regarding $\left( \hat{\tau}_1, \hat{\tau}_2 \right)$-twirling is exactly which two-qubit gates have $\tau$-decompsitions. Most interestingly, are any of these gates not Clifford?
This question can be answered by choosing $\hat{\tau}_1 = \hat{\tau}_2 = \hat{T}$ (where $\hat{T}$ is the T-gate~\cite{devoorhoede_2022}) and $\hat{\mathcal{M}} = \textit{cNOT}$. The resulting $\hat{\mathcal{G}} = (\hat{\tau}_1 \otimes \hat{\tau}_2)^{\dagger} \circ \hat{\mathcal{M}} \circ (\hat{\tau}_1 \otimes \hat{\tau}_2)$ can be confirmed to be non-Clifford by computing $\hat{\mathcal{G}}^{\dagger} \circ \left( \hat{X} \otimes \hat{I} \right) \circ \hat{\mathcal{G}}$ and seeing it is not a two-qubit Pauli string.

\subsection{QAP for Circuits with tau-Decomposible Two-Qubit Gates}

In this section, we show that every condition of the generalized QAP in Sec.~\ref{FormalDefSubSec} can be satisfied -- assuming arbitrary CPTP error -- by trap and target circuits using only a $\hat{\tau}$-decomposible two-qubit gate. 
To that end, we develop target and trap circuits that use $\hat{\tau}$-decomposible two-qubit gates as the only two-qubit gate.

This implies that the standard algorithm (Algorithm~\ref{StandardAccAlg}) can be used to generate an accreditation protocol -- via Theorem~\ref{definitionWorksTheorem} -- in any circumstance where the error in two-qubit gates with $\hat{\tau}$-decompositions can be twirled and the assumptions of Theorem~\ref{definitionWorksTheorem} are met.

These traps use the fact that the two-qubit gates are $\hat{\tau}$-decomposible to express \emph{all} the two-qubit gates in the circuit as Clifford two-qubit gates surrounded by non-Clifford single-qubit gates. 
This allows all these non-Clifford single-qubit gates to effectively cancel -- in the trap circuits -- with the non-Clifford single-qubit gates of neighboring two-qubit gates.
This is depicted in Fig.~\ref{fig:CancellingTau}. The resulting trap circuits are then akin to the trap circuits in Ref.~\cite{Ferracin_2021}.

\begin{figure*}%[h!]
    \centering
    \includegraphics[width=\linewidth]{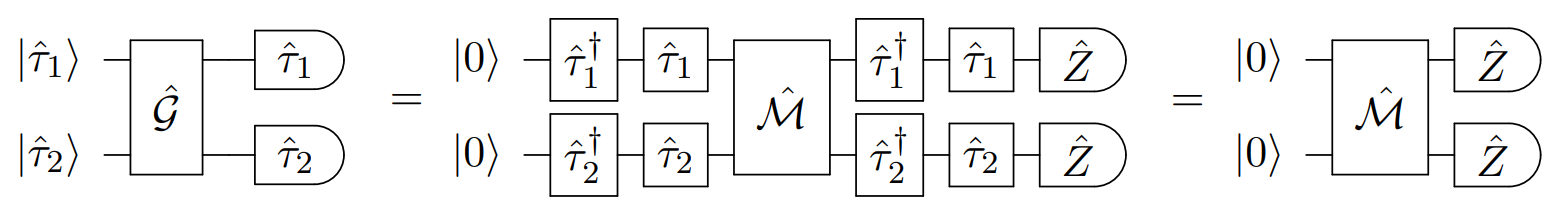}
    \caption{Demonstration of how $\hat{\tau}_1$ and $\hat{\tau}_2$ gates \emph{effectively} cancel – in the ideal case – throughout both the trap and target circuits in our new QAP, to
effectively reduce the, to the trap and target circuits in Ref.~\cite{Ferracin_2021}.}
    \label{fig:CancellingTau}
\end{figure*}

\subsubsection{Main Differences From Ref.~[17] and Introducing Delta Gates}
\label{MainDiffWithFerracin}

In the ideal -- but not likely -- case, where all the two-qubit gates are surrounded immediately by the same two-qubit gate acting on the same qubits, the trap circuits of this new protocol behave exactly as in the Protocol in Ref.~\cite{Ferracin_2021} does. In this case all the $\hat{\tau}_1$ and $\hat{\tau}_2$ gates in the trap circuits cancel with each other to leave exactly the trap circuits from the Protocol in Ref.~\cite{Ferracin_2021}, as in Fig.~\ref{surroundError}.

This cancellation then, as in the Protocol in Ref.~\cite{Ferracin_2021}, leaves a trap circuit that consists entirely of cNOT gates acting on either $\vert 0 \rangle^{\otimes N}$ or $\vert + \rangle^{\otimes N}$. In either case the cNOT gates leave the state invariant and so measurements in the same basis as the initial state give a known deterministic output, if no error occurs. If error does occur, it is effectively transformed to stochastic error and, for at least one of the possible initial states ($\vert 0 \rangle^{\otimes N}$ or $\vert + \rangle^{\otimes N}$), the expected outcome is not obtained, and so the error is detected.

Typically, we are not this fortunate. The cancellation of all the $\hat{\tau}_1$ and $\hat{\tau}_2$ gates does not occur automatically and we are in a scenario not covered by Ref.~\cite{Ferracin_2021}. Progress requires additional action.

To see exactly why, consider situations such as in Fig.~\ref{fig:whyNeedDelta}, where $\hat{\mathcal{G}}$ is any gate with a $\hat{\tau}$-decomposition. 
\begin{figure*}
    \centering 
    \includegraphics[width=\textwidth]{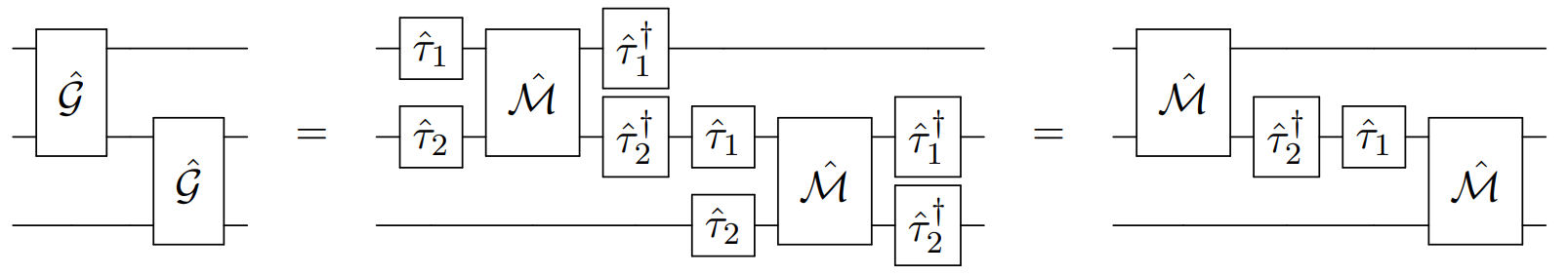}
    \caption[Example non-cancellation of $\hat{\tau}$ gates]{Example non-cancellation of $\hat{\tau}$ gates (assuming all single qubit gates not trapped in the middle qubit are free to leave the shown subcircuit and cancel). To enable these circuit to effectively reduce to the circuits from the Protocol in Ref.~\cite{Ferracin_2021}, additional intervention is required (using $\hat{\Delta}$ and/or $\hat{\Delta}^{\dagger}$ gates).} \label{fig:whyNeedDelta}
\end{figure*}
Fig.~\ref{fig:whyNeedDelta} demonstrates a situation that may arise and result in that $\tau$ gates not cancelling: a $\hat{\tau}_k$ gate meets a $\hat{\tau}^{\dagger}_j$ (where $j \not = k$). These gates do not generally cancel with each other. This situation will happen when two sequential two-qubit gates to act on a specific qubit each also act on another, different (from each other), qubit. This causes the mis-match between $\hat{\tau}_1$ and $\hat{\tau}_2$ that means they do not cancel.

To resolve these problems, we introduce a new single-qubit gate, $\hat{\Delta}$, at each site where two non-canceling $\hat{\tau}_1$ and $\hat{\tau}_2$ gates would otherwise meet. Technically a new $\hat{\Delta}$ gate must be defined for different $\hat{\tau}_1$ and $\hat{\tau}_2$ but the construction of the corresponding $\hat{\Delta}$ is simple and so, for convenience, we act as if $\hat{\tau}_1$ and $\hat{\tau}_2$ are fixed (as they are within a single use of this protocol, as the non-Clifford two-qubit gates entirely and uniquely determine them) and just define $\hat{\Delta}$ in terms of $\hat{\tau}_1$ and $\hat{\tau}_2$ in Def.~\ref{def:DeltaDef}.
\begin{definition}
\label{def:DeltaDef}
For any given $\hat{\tau}_1, \hat{\tau}_2  \in \mathbb{U}(2)$, we define the gate, $\hat{\Delta} \in \mathbb{U}(2)$, by:
\begin{align}
    \hat{\Delta} = \hat{\tau}_1^{\dagger} \hat{\tau}_2.
\end{align}
\end{definition}
An important property of these $\hat{\Delta}$ gates -- and the entire reason we have defined them -- is how they allow $\hat{\tau}_1$ and $\hat{\tau}_2$ gates to cancel with each other. This means that inserting $\hat{\Delta}$ gates (specific to the relevant $\hat{\tau}_1$ and $\hat{\tau}_2$ gates) into a circuit resolves the problem with the gates not canceling automatically, discussed earlier.  

This property of $\hat{\Delta}$ gates is formally stated and proven in Lemma~\ref{DeltaCommute}.
\begin{lemma} $\forall \hat{\tau}_1, \hat{\tau}_2 \in \mathbb{U}(2)$:
    \label{DeltaCommute}
    \begin{align}
        \hat{\tau}_1 \hat{\Delta} \hat{\tau}_2^{\dagger} = \hat{I}
        \textit{ and } \hat{\tau}_2 \hat{\Delta}^{\dagger} \hat{\tau}_1^{\dagger} = \hat{I}. \label{secondinDeltaCommute}
    \end{align}
\end{lemma}
\begin{proof}
        \begin{align}
       \hat{\tau}_1 \hat{\Delta} \hat{\tau}_2^{\dagger}
       =
       \hat{\tau}_1 \left( \hat{\tau}_1^{\dagger} \hat{\tau}_2 \right) \hat{\tau}_2^{\dagger}
       =
       \left( \hat{\tau}_1 \hat{\tau}_1^{\dagger} \right) \left( \hat{\tau}_2 \hat{\tau}_2^{\dagger} \right)
       =
       \hat{I}.
       \end{align}
    The second equation in Lemma~\ref{DeltaCommute} follows from the first via Hermitian conjugation.
\end{proof}
Lemma~\ref{DeltaCommute} enables us to remove the issue of the $\hat{\tau}_1$ and $\hat{\tau}_2$ gates not cancelling in Fig.~\ref{fig:whyNeedDelta}: by placing a $\hat{\Delta}$ or $\hat{\Delta}^{\dagger}$ gate where $\hat{\tau}_1$ and $\hat{\tau}_2$ gates would otherwise meet, the $\hat{\tau}_1$ and $\hat{\tau}_2$ gates instead surround the $\hat{\Delta}$ gate (or $\hat{\Delta}^{\dagger}$ gate, depending on circumstance) to create a situation as in the the left-hand-side of an equality in Eqn.~\eqref{secondinDeltaCommute}, Lemma~\ref{DeltaCommute} then implies that the $\hat{\tau}_1$, $\hat{\tau}_2$ and $\hat{\Delta}$ gates (and their Hermitian conjugates) all cancel out.
Lemma~\ref{DeltaCommute} also shows that the $\hat{\Delta}^{\dagger}$ does the same job when the order of the $\hat{\tau}_1$ and $\hat{\tau}_2$ gates is different, as in Fig.~\ref{fig:DeltaDaggerUsed}.
\begin{figure}
    \centering
    \includegraphics[width=0.5\textwidth]{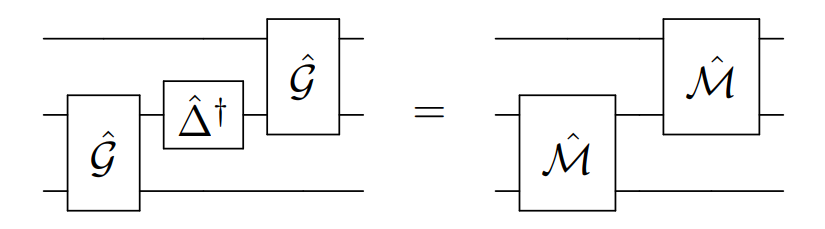}
    \caption{Example cancellation of $\hat{\tau}$ gates, using $\hat{\Delta}^{\dagger}$ and Lemma~\ref{DeltaCommute} (assuming all single qubit gates on either end of wires are free to leave the subcircuit and cancel).} \label{fig:DeltaDaggerUsed}
\end{figure}

Hence, once the $\hat{\Delta}$ gates are added, the ideal circuits behave exactly how they do in the Protocol in Ref.~\cite{Ferracin_2021} as once the $\hat{\tau}_1$, $\hat{\tau}_2$ and $\hat{\Delta}$ gates all cancel the circuits all reduce to exactly the circuits in the Protocol in Ref.~\cite{Ferracin_2021}.

However, these $\hat{\Delta}$ gates have a useful effect in the erroneous case too: allowing error to propagate exactly how it does in the Protocol in Ref.~\cite{Ferracin_2021}, once it is twirled (only a different twirl is used in the Protocol in Ref.~\cite{Ferracin_2021}).\\

Once error occurs and is twirled, it can be viewed as stochastic $\Gamma_{\hat{\tau}_1} \otimes \Gamma_{\hat{\tau}_2}$ error. This can cause problems as $\Gamma_{\hat{\tau}_1} \otimes \Gamma_{\hat{\tau}_2}$ produces stochastic error from different sets on different qubits, i.e., in general, $\Gamma_{\tau_1} \not = \Gamma_{\tau_2}$. This can cause issues as in Fig.~\ref{fig:whyNeedDeltaErrorProp}. Fortunately, this problem is similar to the one $\hat{\Delta}$ and $\hat{\Delta}^{\dagger}$ gates were created to solve and has the same cure.

In some fortunate situations the error propagation works exactly as in the Protocol in Ref.~\cite{Ferracin_2021}. For example assume a $\hat{\tau}^{\dagger}_1 \hat{Z} \hat{\tau}_1$ error occurs on the first qubit in the first $\hat{\mathcal{G}}$ gate (additionally assuming, for convenience that the $\hat{\mathcal{M}}$ corresponding to this choice of $\hat{\mathcal{G}}$ is a $c\hat{Z}$ gate). So it can be seen that the commutation between the error and any $\tau$-decomposible gate is entirely determined by the commutation between the corresponding Pauli error (in this case, $\hat{Z}$) and the relevant Clifford gate (in this case, $c\hat{Z}$).
However, in situations like as in Fig.~\ref{fig:whyNeedDeltaErrorProp}, the propagation of the stochastic error goes awry and the erroneous gates do not remain elements of $\Gamma_{\mathcal{G}}$ as they propagate.
\begin{figure*}
    \centering
        \includegraphics[width=0.75\textwidth]{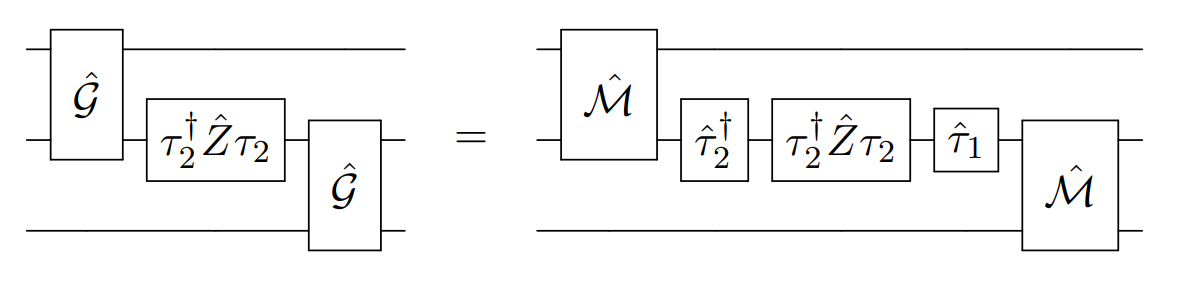}
    \caption{Example where the error fails to propagate due to a mis-match in the $\hat{\tau}$ gate and its inverses after the $\hat{Z}$ gate.} \label{fig:whyNeedDeltaErrorProp}
\end{figure*}
This problem is already resolved by exactly the same $\hat{\Delta}$ gates, in exactly the same position as were added to resolve the $\tau$ gates not cancelling, in Fig.~\ref{fig:whyNeedDelta}. This is due to Lemma~\ref{errorChangeLemma}.

\begin{lemma}
\label{errorChangeLemma}
   If $\hat{\mathcal{P}}$ is any single-qubit Pauli gate, then $\forall \hat{\tau}_1, \hat{\tau}_2 \in \mathbb{U}(2)$:
    \begin{align}
        \hat{\Delta} \hat{\tau}_2^{\dagger} \hat{\mathcal{P}} \hat{\tau}_2 
        =
        \hat{\tau}_1^{\dagger} \hat{\mathcal{P}} \hat{\tau}_1 \hat{\Delta}
        \text{ and } \hat{\Delta}^{\dagger} \hat{\tau}_1^{\dagger} \hat{\mathcal{P}} \hat{\tau}_1 
        =
        \hat{\tau}_2^{\dagger} \hat{\mathcal{P}} \hat{\tau}_2 \hat{\Delta}^{\dagger}.
    \end{align}
Thus pushing an element of $\Gamma_{\hat{\tau}_1}$ through $\hat{\Delta}$ maps it to the corresponding\footnote{Meaning the one with the same Pauli gate in the centre.} element in $\Gamma_{\hat{\tau}_2}$ and pushing an element of $\Gamma_{\hat{\tau}_2}$ through $\hat{\Delta}^{\dagger}$ maps it to the corresponding element in $\Gamma_{\hat{\tau}_1}$.
\end{lemma}
\begin{proof}
        \begin{align}
        \hat{\Delta} \hat{\tau}_2^{\dagger} \hat{\mathcal{P}} \hat{\tau}_2 
        =
        \hat{\tau}_1^{\dagger} \hat{\tau}_2 \hat{\tau}_2^{\dagger} \hat{\mathcal{P}} \hat{\tau}_2
        =
        \hat{\tau}_1^{\dagger} \hat{\mathcal{P}} \hat{\tau}_1 \hat{\tau}_1^{\dagger} \hat{\tau}_2
        =
        \label{finalLineProvingFirstrrorChangeLemma}
        \hat{\tau}_1^{\dagger} \hat{\mathcal{P}} \hat{\tau}_1 \hat{\Delta}.
        \end{align}
   The second equation follows from the first via Hermitian conjugation.
\end{proof}
Hence, if Fig.~\ref{fig:whyNeedDeltaErrorProp} had a $\hat{\Delta}$ gate, as is already required (but not present in Fig.~\ref{fig:whyNeedDeltaErrorProp}) for the $\hat{\tau}$ gates to cancel in the ideal case (i.e., between the two $\hat{G}$ gates, on the middle qubit), then the error would propagate exactly as in the Protocol in Ref.~\cite{Ferracin_2021}.

As before, in the error-free case, $\hat{\Delta}^{\dagger}$ performs a similar task to $\hat{\Delta}$, when the order of $\hat{\tau}_1$ and $\hat{\tau}_2$ is reversed (i.e., when they switch places). 

In conclusion, while the error itself may be different; it propagates in exactly the same way as the corresponding Pauli error does in the Protocol in Ref.~\cite{Ferracin_2021}.

\subsubsection{Defining Trap and Target Circuits}
\label{sec:trapTarget}
As mentioned earlier, the only way the protocol of this paper differs from the Protocol in Ref.~\cite{Ferracin_2021} is in the trap and target circuits, \emph{all else is exactly the same}.
This section is where these changes are detailed and we present the trap and target circuits to use in our protocol.

We start by defining the simpler of the two, the target circuits, in Algorithm \ref{targBuildingAlg}.
This algorithm uses the notion of a $\tau$-basis, which we define in Def.~\ref{def:taubasis}.
\begin{definition}
    \label{def:taubasis}
    We define a specific \underline{$\tau$-basis}, $\beta_{\hat{\tau}}$, via a specific single-qubit gate, $\hat{\tau} \in \mathbb{U}(2)$. It is defined as: $\beta_{\hat{\tau}} = \{ \vert \hat{\tau} \rangle, \vert \lnot \hat{\tau} \rangle \}$, such that:
    \begin{align}
        \vert \hat{\tau} \rangle &= \hat{\tau} \vert 0 \rangle\\
       \text{ and } \vert \lnot \hat{\tau} \rangle &= \hat{\tau} \vert 1 \rangle.
    \end{align}
\end{definition}

\begin{figure}
    \centering
\begin{algorithm}[H]
%\SetAlgoLined
$\mathbf{Input:}$ \\
$\bullet$ A circuit, $\mathcal{C}$\\
$\bullet$ $\hat{\tau}_1$, $\hat{\tau}_2$ from the two-qubit gate's $\hat{\tau}$-decomposition
\begin{enumerate} 
     \item newTarg = $\mathcal{C}$
    
     \item \For{Qubit with initial state $\vert i \rangle$}{
     \begin{enumerate} 
    \item Change the initial state of the qubit to $\vert \hat{\tau}_1 \rangle$ (defined via Def.~\ref{def:taubasis})\\
    \item Add a gate mapping $\vert \hat{\tau}_1 \rangle$) to the original input state, $\vert i \rangle$ immediately after the state preparation.
    \end{enumerate}
    }
     \item \For{Measurment in $\vert i \rangle$ basis}{
     \begin{enumerate} 
    \item Change the measurement basis to $\beta_{\hat{\tau}_1}$
    \item  Before the measurement add a single-qubit gate mapping from the old measurement basis to the $\beta_{\hat{\tau}}$ basis
    \end{enumerate}
    }
     \item $\left( \hat{\tau}_1, \hat{\tau}_2 \right)$-\textit{twirl, as in Def.}~\ref{twirlDef} \textit{, every } $\mathcal{G}$ \textit{ gate in } newTarg
    
     \item With probability $0.5$:
     \begin{enumerate} 
     \item Add $\hat{\tau}_1^{\dagger} \hat{Z} \hat{\tau}_1$ on each qubit after state preparation of newTarg\\
 \item Add $\hat{\tau}_1^{\dagger} \hat{Z} \hat{\tau}_1$ before each measurement of newTarg
\end{enumerate}
\end{enumerate}
$\mathbf{Return}:$ newTarg
\caption{Algorithm to obtain a single target circuit
 \label{targBuildingAlg}}
\end{algorithm}
\end{figure}

With the target circuits established, we define the trap circuits via Algorithm \ref{trapBuildingAlg} to construct the corresponding traps from a given input circuit.

The final step is to prove that the given targets and traps provide a QAP, as in Theorem~\ref{accreditationTheorem}.

\begin{figure}[h!]
    \centering
\begin{algorithm}[H]
%\SetAlgoLined
$\mathbf{Input:}$ \\
$\bullet$ A circuit, $\mathcal{C}$\\
$\bullet$ $\hat{\tau}_1$, $\hat{\tau}_2$ from the two-qubit gate's $\hat{\tau}$-decomposition
 \begin{enumerate} 
    \item newTrap = $\mathcal{C}$

    \item Remove all single-qubit gates from newTrap, replacing them with identity gates
     
    \item Apply Algorithm \ref{targBuildingAlg} to newTrap
    
    \item Add $\hat{\Delta}$ and $\hat{\Delta}^{\dagger}$ gates to newTrap as per Algorithm \ref{addDeltaAlg}
    
    \item With probability $0.5$, add $\hat{\tau}_1^{\dagger} H \hat{\tau}_1$ immediately before measurement and after state preparation to newTrap.
\end{enumerate}
$\mathbf{Return}:$ newTrap
\caption{Algorithm to obtain a single trap circuit
 \label{trapBuildingAlg}}
\end{algorithm}
\end{figure}
\begin{figure}[h!]
    \centering
\begin{algorithm}[H]
%\SetAlgoLined
$\mathbf{Input:}$ \\
$\bullet$ A circuit, $\mathcal{C}$\\
$\bullet$ Descriptions of $\hat{\tau}_1$ and $\hat{\tau}_2$\\
 \begin{enumerate}
 
 \item \For{Two-qubit gate, $g$, in $\mathcal{C}$}{
 \begin{enumerate}
    \item \For{Qubit, Q, the gate, $g$, acts on}{
    \begin{enumerate}
        \item  index[$g$, Q] = which of $\hat{\tau}_1$ or $\hat{\tau}_2$ acts on Q in the $\hat{\tau}$-decomposition of $g$
         \end{enumerate}
    }
     \end{enumerate}
 }
 
\item \For{Gate, $g$, in $\mathcal{C}$}{
\begin{enumerate}
    \item \For{Qubit, Q, the gate acts on}{
    \begin{enumerate} 
    \item Follow Q after $g$ until another $\mathcal{G}$ gate, $g2$, acts on Q

    \item \If{index[$g$, Q] $ \not =$ index[$g2$, Q]}{
    \begin{enumerate}
        \item  Calculate which of $\hat{\Delta}$ or $\hat{\Delta}^{\dagger}$ maps index[$g$, Q] to index[$g2$, Q]
        \item Add the above chosen gate to qubit Q after $g$
        \end{enumerate}
        }
    \end{enumerate}
    }
    \end{enumerate}
 }

 \item \For{Qubit, Q, in the circuit}{
 \begin{enumerate} 
    \item Let gate, $g_{\mathrm{last}} (Q)$ be the last gate to act on qubit $Q$ before measurement
    \item \If{index[$g$, Q] $ \not =$ $\hat{\tau}_1$}{
    \begin{enumerate}
        \item Add a $\hat{\Delta}$ gate immediately before measurement
    \end{enumerate}
    }

 \end{enumerate} }
\end{enumerate}
$\mathbf{Return}:$ Updated $\mathcal{C}$
\caption{Algorithm to allow error propagation in traps
 \label{addDeltaAlg}}
\end{algorithm}
\end{figure}

\begin{theorem}
    \label{accreditationTheorem}
    Assuming $\mathbf{N1}$ and $\mathbf{N2}$, any circuit where the only type of two-qubit gate, $\hat{\mathcal{G}}$, is $\hat{\tau}$-decomposable can be accredited.
\end{theorem}
\begin{proof}
We proceed via Theorem~\ref{definitionWorksTheorem}, showing the conditions required for it to apply (that is, the conditions in Def.~\ref{accreditationFormalDEf}) are met. This then implies that, using Algorithm~\ref{StandardAccAlg}, the trap and target circuits generated by Algorithm~\ref{trapBuildingAlg} and Algorithm~\ref{targBuildingAlg}, respectively, successfully accredit circuit executions.\\

\noindent \underline{Conditions on the target circuit}
\begin{enumerate}
\item [$\mathbf{1a:}$] The target circuit can be created in time linear in the number of gates in the input circuit, by Algorithm~\ref{targBuildingAlg}.
\item [$\mathbf{1b:}$] The number of gates in the target can be seen -- in Algorithm~\ref{targBuildingAlg} -- to be linear in the size of the input circuit.
\item [$\mathbf{1c:}$] The number of qubits in the target circuit is exactly the same as in the input circuit: Algorithm~\ref{targBuildingAlg} adds no qubits.
\item [$\mathbf{1d:}$] Algorithm~\ref{targBuildingAlg} produces a circuit with exactly the same probabilities of each measurement outcome as in the input circuit, in the ideal case -- as shown in Lemma~\ref{lem:SameInIdealCase}.
\item [$\mathbf{1e:}$] The twirling in Algorithm~\ref{targBuildingAlg} reduces all error in the execution of the target circuit to $\Gamma_{\hat{\mathcal{G}}}$-stochastic error, as in Lemma~\ref{toPauli}.
\end{enumerate}

\noindent \underline{Conditions on the trap circuits}
\begin{enumerate}
 \item [$\mathbf{2a:}$]  The required trap circuits can be created in time linear in the number of gates in the input circuit, by Algorithm~\ref{trapBuildingAlg}.
\item [$\mathbf{2b:}$] The number of gates in any trap can be seen -- in Algorithm~\ref{trapBuildingAlg} -- to be linear in those of the input circuit.
\item [$\mathbf{2c:}$] The number of qubits in the trap circuits is exactly the same as in the input circuit: Algorithm~\ref{trapBuildingAlg} adds no qubits.
\item [$\mathbf{2d:}$] Lemma~\ref{simulable} shows the output of any trap circuit not experiencing error is fixed and classically computable.
\item [$\mathbf{2e-f:}$] Lemma~\ref{toPauli} implies that, presupposing assumptions $\mathbf{N1}$ and $\mathbf{N2}$, all error is equivalent to stochastic $\Gamma_{\hat{\tau}_1}$ error. Lemma~\ref{allDetectLemma} then implies that all stochastic $\Gamma_{\hat{\tau}_1}$ error occurring in an execution of a trap circuit is detected with probability of at least $0.5$ as the layers of $\hat{\tau}^{\dagger}_1 H \hat{\tau}_1$ gates are present with probability $0.5$.\\
\end{enumerate}
\noindent \underline{Relation between trap and target error}
\begin{enumerate}
    \item [$\mathbf{3:}$] As trap and target circuits only differ in their single-qubit gates, $\mathbf{N2}$ implies that the error in trap and target circuit executions is identical.
\end{enumerate}
\end{proof}
The proof of Theorem~\ref{accreditationTheorem} has made use of Lemmas~\ref{simulable},~\ref{lem:SameInIdealCase},~\ref{toPauli},~and~\ref{allDetectLemma} which are proven in Appendix~\ref{sec:theorem3usedLemmas}.

\section{Twirling and Accrediting XY-Interaction Gates}
\label{XYDecSec}

We now move to applying our QAP to gates that more closely reflect those used in contemporary quantum hardware.  These are based on the
XY interaction and are in extensive use today (and can be performed with high fidelity~\cite{Abrams2020}): in part because they are a subset of the fSim gate~\cite{PhysRevLett.125.120504, Arute2019, Nguyen_2024}.
Additionally, any XY-interaction gate can be performed using a single iSWAP-like gate or CPHASE gate already implemented at Google~\cite{PhysRevLett.125.120504}, and can be implemented using a single Givens gate~\cite{Kivlichan_2018} surrounded by single-qubit gates.

We begin by defining the XY-interaction gate~\cite{Schuch_2003}.
\begin{definition}
    \label{def:XYinteractGate}
    A gate, $\hat{\mathcal{G}} \in \mathbb{U}(4)$, is an \underline{XY-interaction gate} if it can be expressed, for some value of $t \in \mathbb{R}$, as:
    \begin{align}
    \label{eXYGateDef}
    \hat{\mathcal{G}} = \exp\left(-i t \left( \hat{X}_1 \hat{X}_2 + \hat{Y}_1 \hat{Y}_2 \right)\right).
\end{align}
\end{definition}
If such a $t \in \mathbb{R}$ exists for a gate, $\hat{\mathcal{G}}$, we refer to Eqn.~\eqref{eXYGateDef} as the XY-decomposition of $\hat{\mathcal{G}}$. 
As with $\tau$-decompositions, not every gate has an XY-decomposition but we refer to those that do as XY-decomposable.

\subsection{Twirling XY-Interaction Gates}
\label{subsec:twirlingXYInters}

We first ask if error in these gates can be twirled to stochastic Pauli error. This begins with a basic consideration of the commutation properties of XY-interaction gates, in Lemma~\ref{commutingLemmaRestate}.
\begin{lemma}
    \label{commutingLemmaRestate}
    Let $\hat{H} = \hat{X}_1 \hat{X}_2 + \hat{Y}_1 \hat{Y}_2$, then
     $\forall \Bar{\mathfrak{B}}_j \in \{ \hat{X}^{\otimes 2}, \hat{Y}^{\otimes 2}, \hat{Z}^{\otimes 2} \}$ and $\forall t \in \mathbb{R}$,
    \begin{align}
        \left [ \Bar{\mathfrak{B}}_j, \exp \left( - i \hat{H} t /2 \right)\right ] = 0.
    \end{align}
\end{lemma}
\begin{proof}
    For any choice of $\Bar{\mathfrak{B}}_j \in \{ \hat{X}^{\otimes 2}, \hat{Y}^{\otimes 2}, \hat{Z}^{\otimes 2} \}$, $\Bar{\mathfrak{B}}_j$ commutes with both $\hat{X}_1 \hat{X}_2$ and $\hat{Y}_1 \hat{Y}_2$ as all Pauli gates either commute or anti-commute: if two Pauli gates anti-commute then if we take the tensor product of each with itself then these products commute. Therefore, $\forall \Bar{\mathfrak{B}}_j \in \{ \hat{X}^{\otimes 2}, \hat{Y}^{\otimes 2}, \hat{Z}^{\otimes 2} \}$ and $\forall t \in \mathbb{R}$,
    \begin{align}
        \left [ \Bar{\mathfrak{B}}_j, \exp \left( - i \hat{H} t /2 \right)\right ] = 0.
    \end{align}
\end{proof}
Therefore, for an XY-interaction gate, $\hat{\mathcal{G}}$, expressible as in Eqn.~\eqref{eXYGateDef}, if $\hat{\mathcal{G}}$ is implemented erroneously, with CPTP error, i.e., $\Tilde{G} = \mathcal{E} \circ \hat{\mathcal{G}}$ (again, as in Fig.~\ref{surroundError}, we note that this is an abuse of notation but it is done for readability and just means apply the unitary $\hat{\mathcal{G}}$ followed by the CPTP map $\mathcal{E}$) then,  $\forall \hat{P} \in \mathbb{P}_1$, the circuit equation in Fig.~\ref{surroundErroreXY} holds.
\begin{figure*}
    \centering
    \includegraphics[]{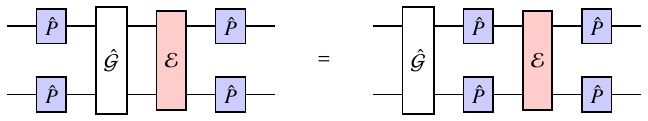}
\caption{ $\hat{P} \otimes \hat{P} \in \mathbb{P}_1^{\otimes 2}$ commuting through a gate, $\hat{\mathcal{G}}$, with an XY-decomposition, to surround CPTP error, $\mathcal{E}$.}
\label{surroundErroreXY}
\end{figure*}

Fig.~\ref{surroundErroreXY} is similar to Fig.~\ref{surroundError} as they both demonstrate how their respective methods to twirl error occurring in gates manages to surround the CPTP error with the required single-qubit gates.
As a result of Fig.~\ref{surroundErroreXY}, any XY-twirlable error (defined in Def.~\ref{XYTwirlableDef}) occurring in the two-qubit gate, $\hat{\mathcal{G}}$, can be twirled to stochastic Pauli error, as stated in Theorem~\ref{TwoFoldTwirlTheorem} which is proven in Appendix~\ref{ProofXyTwirlTheoremAppendix}.
\begin{definition}
    \label{XYTwirlableDef}
    CPTP error is \underline{XY-twirlable} if its Kraus operators all have a Pauli decomposition such that all terms with non-zero coefficients are exclusively from a set, $\mathbb{K}$, such that:
    \begin{align}
        &\bullet \mathbb{K} \subset \mathbb{P}_1^{\otimes 2}\\
        &\bullet \forall \hat{P}, \hat{Q} \in \mathbb{P}_1 , \hat{P} \otimes \hat{Q} \in \mathbb{K} \Rightarrow \hat{Q} \otimes \hat{P}  \not \in \mathbb{K}.
    \end{align}
    An exception to the second condition is made for $\hat{I} \otimes \hat{I}$. We refer to $\mathbb{K}$ as an \underline{unflippable set}. Examples of unflippable sets include:
    \begin{enumerate}
        \item $\big \{ \hat{I} \otimes \hat{I}, \hat{X} \otimes \hat{I},\hat{Y} \otimes \hat{I}, \hat{Z} \otimes \hat{I} \big \}$
        \item $\big \{ \hat{I} \otimes \hat{I}, \hat{X} \otimes \hat{Y}, \hat{Z} \otimes \hat{Y}, \hat{Z} \otimes \hat{X}, \hat{I} \otimes \hat{X} \big \}$
        \item $\big \{  \hat{Z} \otimes \hat{Y}, \hat{Z} \otimes \hat{X}, \hat{I} \otimes \hat{X}, \hat{Y} \otimes \hat{I} \big \}$
    \end{enumerate}
\end{definition}
\begin{theorem}
    \label{TwoFoldTwirlTheorem}
    All XY-twirlable errors occurring in a XY-decomposable gate can be twirled to stochastic Pauli error by conjugating it with Pauli gates (including the identity) chosen uniformly at random.
\end{theorem}
In fact we can further expand the error that can be twirled to stochastic error to any stochastic combination of XY-twirlable error, as in Corollary~\ref{stochXYError}.
\begin{corollary}
    \label{stochXYError}
    Any stochastic combination of XY-twirlable error occurring in a XY-decomposable gate can be twirled to stochastic Pauli error by adding only single-qubit Pauli gates.
\end{corollary}
\begin{proof}
    Consider the case where there are two distinct XY-twirlable errors, $\mathcal{E}_1$ and $\mathcal{E}_2$, each occurring with their respective probabilities $\mathrm{R}_1, \mathrm{R}_2 \in [0,1]$ after a two-qubit gate, $\hat{g} \in \mathbb{U}(4)$.
    
    Separately consider the cases where each XY-twirlable error occurs: each case is established to be effectively twirlable into to stochastic Pauli error, as per Theorem \ref{TwoFoldTwirlTheorem}. Hence the error in $\hat{g}$ can be considered as:
    \begin{align}
        \sum_{j = 0}^{3} \left( 
 \alpha_j^{(\mathcal{E}_1)} \hat{\sigma}_j \hat{g} \rho \hat{g}^{\dagger} \left( \hat{\sigma}_j \right)^{\dagger} \right)
 \hspace{1em} \textit{or} \hspace{1em}
 \sum_{j = 0}^{3} \left( 
 \alpha_j^{(\mathcal{E}_2)} \hat{\sigma}_j \hat{g} \rho \hat{g}^{\dagger} \left( \hat{\sigma}_j \right)^{\dagger} \right),
    \end{align}
    where $\rho$ is any state of the qubits $g$ acts on, and $\alpha_j^{(\mathcal{E}_k)} = \dfrac{1}{2^{2N}} \big \vert \text{Tr}\left( \hat{\sigma}_j \mathcal{E}_k\right) \big \vert^2$.
    Then considering each of these errors applied stochastically, the effective overall state after the gate will be:
    \begin{align}
    \rho^{\prime}
    &=
    \label{finalFormMixedError}
    \sum_{j = 0}^{3} \left( \bigg[ \left( 1 - \mathrm{R}_1 - \mathrm{R}_2 \right) \delta_{j, 0} + \mathrm{R}_1 \alpha_j^{(\mathcal{E}_1)} +
    \mathrm{R}_2 \alpha_j^{(\mathcal{E}_2)} \bigg] \hat{\sigma}_j \hat{g} \rho \hat{g}^{\dagger} \hat{\sigma}_j^{\dagger} \right).
    \end{align}
Eqn.~\eqref{finalFormMixedError} can be interpreted as the state resulting from the gate, $\hat{g}$, being applied correctly and then experiencing stochastic Pauli error. Hence the overall error is effectively reduced to stochastic Pauli error, as required.
\end{proof}

\subsection{Twirling Errors Not in Unflippable Set: Extra Assumption}
\label{subsec:ImprovedXYTwirling}

We now wish to avail of the features of twirling for errors from outside the unflippable set. This seems to be impossible without additional assumptions.
One such assumption, originating in Ref.~\cite{carrasco2023gaining} (but expressed differently here), allows for % this:
\begin{enumerate}
    \item[$\mathbf{N3:}$] \label{ExtraAss} Define $\mathbb{H} = \big \{ \hat{H}'_{s_1, s_2} \text{ } \vert \text{ } s_1, s_2 \in \{+1, -1\} \big \}$, where:
    \begin{align}
    \hat{H}'_{s_1, s_2}
    =
    s_1 \hat{X}_1 \hat{X}_2 + s_2 \hat{Y}_1 \hat{Y}_2.
    \end{align}
Then, for any fixed value of $t \in \mathbb{R}$, all gates in the set
$ \{\exp(-i \hat{H}' t) \text{ } \vert \text{ } \hat{H}' \in \mathbb{H} \}$ experience identical error. 
\end{enumerate}

This means that the error in distinct implementations of any of the gates in the set is represented by the exact same CPTP map. Examples of real-world error that conforms to this requirement are some stray couplings (e.g. an extra $\hat{Z}_1 \hat{Z}_2$ term in $\hat{H}'_{s_1, s_2}$) \cite{Xu_2024}, and fluctuations -- independent of the intended values of $s_1$ and $s_2$ -- in the values of $s_1$ and $s_2$ (in $\hat{H}'_{s_1, s_2}$) as the time evolution according to $\hat{H}'_{s_1, s_2}$ is being applied.
Twirling of any error subject to \textbf{N1} - \textbf{N3} in XY-interaction gates is now possible due to Lemma~\ref{extraAssumptionXYTwirlLemma}.
It first requires Def.~\ref{XIPDef}.

\begin{definition}
    \label{XIPDef}
    Define the function \underline{$\xi'$}$: \mathbb{P}_1 \times \mathbb{P}_1 \rightarrow \{ +1, -1 \}$ by: $\forall \hat{p}_1, \hat{p}_2 \in \mathbb{P}_1$,
    \begin{align}
         \hat{p}_1 \hat{p}_2 = \xi'(\hat{p}_1, \hat{p}_2) \hat{p}_2 \hat{p}_1.
        \end{align}
        Then define, $\forall \hat{P} \in \mathbb{P}_1$,
    $\Xi(\hat{P}) = \hat{H}'_{\xi'(\hat{X}, \hat{P}), \xi'(\hat{Y}, \hat{P})}$.
\end{definition}

\begin{lemma}
    \label{extraAssumptionXYTwirlLemma}
    Consider a circuit containing arbitrarily many XY-interaction gates
     and assume $\mathbf{N1}-\mathbf{N3}$. If, for each specific instance of $\exp (-i \hat{H} t )$ in each execution (experiencing CPTP error), $\hat{P} \in \mathbb{P}_1^{\otimes 2}$ is chosen uniformly at random and $\hat{P} \exp \left (-i \Xi(\hat{P}) t \right) \hat{P}$ is applied instead; the measurement outcomes -- of those executions -- are the same as if they experienced only stochastic Pauli error.
\end{lemma}

\begin{proof}
    The uniformly random application of $\hat{P} \exp \left ( -i \Xi(\hat{P}) t \right) \hat{P}$, with $\hat{P}$ chosen uniformly from $ \mathbb{P}_1^{\otimes 2}$ to arbitrary density matrix, $\rho$, experiencing arbitrary CPTP error (denoted by $\mathcal{E}$, which is invariant under changing values of $\hat{P}$ due to $\mathbf{N3}$) can be expressed as:
    \begin{align}
        \label{ExtraAssumpTwirlFirstExpress}
        \dfrac{1}{\vert  \mathbb{P}_1^{\otimes 2} \vert}
        \sum_{P \in \mathbb{P}_1^{\otimes 2}}
        \left(
            \hat{P} \mathcal{E} \left( \exp \left ( -i \Xi(\hat{P}) t \right) \hat{P} \rho \hat{P} \exp \left( i \Xi(\hat{P}) t \right) \right) \hat{P}
        \right).
    \end{align}
    Via properties of the Pauli group and the Taylor expansion of the matrix exponential, $\forall \hat{P} \in \mathbb{P}_1^{\otimes 2}$:
    \begin{align}
        \label{useOfXiEquation}
        \exp \left ( -i \Xi(\hat{P}) t \right)
        &=
       \hat{P} \exp \left ( -i \hat{H} t \right) \hat{P}.
    \end{align}
    Substituting Eqn.~\eqref{useOfXiEquation} into Eqn.~\eqref{ExtraAssumpTwirlFirstExpress} and letting $\rho' = \exp \left( -i \hat{H}_I t \right) \rho \exp \left( i \hat{H} t\right)$, we derive:
    \begin{align}
        \label{ExtraAssumpTwirlSecondExpress}
        \dfrac{1}{\vert  \mathbb{P}_1^{\otimes 2} \vert}
        \sum_{\hat{P} \in \mathbb{P}_1^{\otimes 2}}
        \left(
            \hat{P} \mathcal{E} \left( \hat{P} \rho' \hat{P} \right) \hat{P}
        \right).
    \end{align}
    Eqn.~\eqref{ExtraAssumpTwirlSecondExpress} represents arbitrary CPTP error, $\mathcal{E}$, acting on $\rho'$ being Pauli twirled to stochastic Pauli error. This can be seen using Lemma~\ref{MostGeneralTwirlLemma}, in Appendix~\ref{TwirlingTauDeets}, by setting $\mathbb{S} = \Lambda = Q = \mathbb{P}_1^{\otimes 2}$. Therefore, we conclude that all CPTP error (conforming to $\mathbf{N1-3}$) occurring within the two-qubit gate is effectively reduced to stochastic Pauli error.
\end{proof}

\subsection{QAP for Circuits with XY-Interaction Gates}
\label{AccOfXYSec}

With the ability to twirl errors occurring in XY-interaction gates established, we now turn to using this to accredit circuits containing XY-interaction gates. The basis for this is a sub-circuit that we call a $j$-vanishing block defined and studied in Secs.~\ref{vanishingBlockIntro} and~\ref{ErrorVanishingSec}.
In our QAP, the difference between the $j$-vanishing blocks in the same position in a target and a trap is encoded in a binary bit, that we refer to as $j$ and it is from this that $j$-vanishing blocks get their name.
In a trap circuit, $j$ is set to one in every vanishing block, and in a target circuit, $j$ is set to zero in every vanishing block.

\subsubsection{j-Vanishing Blocks}
\label{vanishingBlockIntro}
\begin{definition}
\label{vanishBlockDef}
A \underline{$j$-vanishing block} is a subcircuit of the form in Fig. \ref{fig:vanishingBlockDef}.
\begin{figure*}%[h!]
    \centering
    \includegraphics[]{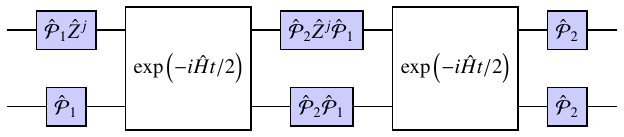}
    \caption{Depiction of a $j$-vanishing block, where $\hat{H} = \hat{X}\hat{X} + \hat{Y} \hat{Y}$, $t \in \mathbb{R}$, and $j \in \{ 0, 1\}$. Each $\hat{\mathcal{P}}_k$ is a random Pauli \footnote{But fixed throughout each circuit execution.} selected according to the twirling scheme (either from Sec.~\ref{subsec:twirlingXYInters} or Sec.~\ref{subsec:ImprovedXYTwirling}) \footnote{We note the Sec.~\ref{subsec:ImprovedXYTwirling} scheme requires $\exp \left(-i \hat{H} t /2 \right)$ be replaced with some other gate but, for simplicity this is neglected in our definition of the vanishing blocks and makes no difference to how we use the vanishing blocks}.
    \label{fig:vanishingBlockDef}}
\end{figure*}
\end{definition}
The $j$ parameter is present only in the single-qubit gates in Fig.~\ref{fig:vanishingBlockDef}.
This makes the sub-circuit different for different values of $j$. 
 We can then proceed to prove the vital properties of $j$-vanishing blocks. A pre-requisite of this is Lemma \ref{InvertingLemma}, which is based on the time inversion circuits found in Ref.~\cite{jackson2023accreditation}.
 \begin{lemma}
\label{InvertingLemma}
    $\forall t \in \mathbb{R}$, if $\hat{H} = \hat{X}_1 \hat{X}_2 + \hat{Y}_1 \hat{Y}_2$,
    \begin{align}
        \left( \hat{Z} \otimes \hat{I} \right) \exp \left( -i \hat{H}t / 2 \right) \left( \hat{Z} \otimes \hat{I} \right)
        &=
        \exp \left( i \hat{H} t / 2 \right).
    \end{align}
\end{lemma}
\begin{proof} Due to the commutation relations of the Pauli group:
    \begin{align}
        \left( \hat{Z} \otimes \hat{I} \right) \hat{H} \left( \hat{Z} \otimes \hat{I} \right)
        &=
        \hat{Z} \hat{X} \hat{Z} \otimes \hat{X} + \hat{Z} \hat{Y} \hat{Z} \otimes \hat{Y} =
        \label{InvertedH}
        - \hat{H}.
    \end{align}
    Then, using Eqn.~\eqref{InvertedH} and the Taylor expansion of $\exp \left( -i \hat{H}t / 2 \right)$,
    \begin{align}
        \left( \hat{Z} \otimes \hat{I} \right) \exp \left( -i \hat{H}t / 2 \right) \left( \hat{Z} \otimes \hat{I} \right)
        &=
        \label{expWithZsIn}
        \exp \left( -i \left( \hat{Z} \otimes \hat{I} \right) \hat{H} \left( \hat{Z} \otimes \hat{I} \right) t / 2 \right) \nonumber \\
        &=
        \exp \left( i \hat{H} t / 2 \right).
    \end{align}
\end{proof}
The most important properties of $j$-vanishing blocks, when implemented correctly, is what operators they are equivalent to for different values of $j \in \{0, 1\}$. These equivalences are established in Lemmas~\ref{IdealJZeroLemma} and \ref{IdealJOneLemma}.
\begin{lemma}
\label{IdealJZeroLemma}
    The error-free implementation of a $0$-vanishing block is equivalent to $\exp \left( -i \hat{H}t  \right)$.
\end{lemma}
\begin{widetext}
\begin{proof}
    A $0$-vanishing block may be expressed mathematically as:
    \begin{align}
        \left(  \hat{\mathcal{P}}_2 \otimes  \hat{\mathcal{P}}_2 \right)\exp \left( -i \hat{H} t/2  \right) \left( \hat{\mathcal{P}}_2 \hat{Z}^0 \hat{\mathcal{P}}_1 \otimes \hat{\mathcal{P}}_2 \hat{\mathcal{P}}_1 \right)\exp \left( -i \hat{H} t/2  \right) \left( \hat{\mathcal{P}}_1 \hat{Z}^0  \otimes \hat{\mathcal{P}}_1  \right).
    \end{align}
    Using the commutation relations between $\exp \left ( -i \hat{H} t \right)$ and Pauli gates, due to Lemma~\ref{commutingLemmaRestate}, this becomes:
    \begin{align}
        &\exp \left( -i \hat{H} t/2 \right) \left(  \hat{\mathcal{P}}_2 \otimes  \hat{\mathcal{P}}_2 \right) \left( \hat{\mathcal{P}}_2 \hat{\mathcal{P}}_1 \otimes \hat{\mathcal{P}}_2 \hat{\mathcal{P}}_1 \right) \left( \hat{\mathcal{P}}_1  \otimes \hat{\mathcal{P}}_1  \right)\exp \left( -i \hat{H} t/2  \right)
        =
        \exp \left ( -i \hat{H} t \right).
    \end{align}
\end{proof}
\begin{lemma}
    \label{IdealJOneLemma}
    The error-free implementation of a $1$-vanishing block is equivalent to the identity.
\end{lemma}
\begin{proof}
    A $1$-vanishing block may be expressed mathematically as:
    \begin{align}
        \left(  \hat{\mathcal{P}}_2 \otimes  \hat{\mathcal{P}}_2 \right)\exp \left( -i \hat{H} t/2  \right) \left( \hat{\mathcal{P}}_2 \hat{Z}^1 \hat{\mathcal{P}}_1 \otimes \hat{\mathcal{P}}_2 \hat{\mathcal{P}}_1 \right)\exp \left( -i \hat{H} t/2  \right) \left( \hat{\mathcal{P}}_1 \hat{Z}^1  \otimes \hat{\mathcal{P}}_1  \right).
    \end{align}
    Using the commutation relations between $\exp \left ( -i \hat{H} t \right)$ and Pauli gates, and Lemma~\ref{InvertingLemma}, this becomes:
    \begin{align}
     \label{DueToLemmaInvertingLemma}
        &\exp \left( -i \hat{H} t/2 \right) \left(  \hat{\mathcal{P}}_2 \hat{\mathcal{P}}_2 \hat{\mathcal{P}}_1 \hat{\mathcal{P}}_1  \otimes  \hat{\mathcal{P}}_2 \hat{\mathcal{P}}_2 \hat{\mathcal{P}}_1 \hat{\mathcal{P}}_1  \right) \left( \hat{Z} \otimes \hat{I} \right) \exp \left( -i \hat{H} t / 2  \right) \left( \hat{Z} \otimes \hat{I} \right) \nonumber \\
        &=
       \exp \left( -i \hat{H} t/2  \right) \left( \hat{Z} \otimes \hat{I} \right) \exp \left( -i \hat{H} t / 2  \right) \left( \hat{Z} \otimes \hat{I} \right)
        =
       \exp \left( -i \hat{H} t/2  \right) \exp \left( i \hat{H} t/2  \right)
        =
        \hat{I}. 
    \end{align}
\end{proof}
\end{widetext}
\subsubsection{Erroneous j-vanishing blocks}
\label{ErrorVanishingSec}
With the behavior of $j$-vanishing blocks in the error-free case established, we now consider their erroneous implementation.
Using an approach similar to Ref.~\cite[Lemma 1]{jackson2023accreditation}, we consider the single-qubit gates in a $j$-vanishing block to be error-free by folding their error into the error of the two-qubit gates of that same $j$-vanishing block.
\begin{lemma}
\label{StochErrorInVanishing}
 Subject to errors conforming to \textbf{N1}-\textbf{N3}, erroneous $j$-vanishing blocks
 act as ideal $j$-vanishing blocks with stochastic error (that is independent of $j$) acting before and after the block -- with the $\hat{\mathcal{P}}_j$ gates removed.
\end{lemma}
\begin{proof}
    First note that the $\hat{\mathcal{P}}_j$ gates twirl error to stochastic error and that the choice to consider CPTP error as acting immediately after the ideal version of the gate experiencing error is entirely arbitrary: it is equally valid for error to be modelled as CPTP maps acting immediately before the ideal version of the gate experiencing error. Using this let $\mathcal{E}_1$ and $\mathcal{E}_2$ be CPTP maps modelling the error in the first and second two-qubit gates, respectively, in the considered block, as in Fig.~\ref{fig:jVanishwithCPTP}.
\begin{figure*}
\centering
\begin{center}
    \includegraphics[]{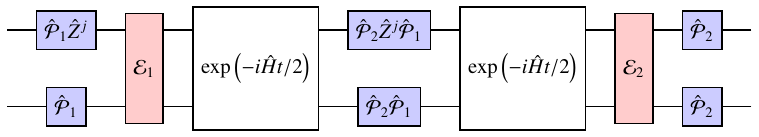}
\end{center}
    \caption{A $j$-vanishing block experiencing CPTP error ($\mathcal{E}_1$ and $\mathcal{E}_2$) and the $\hat{\mathcal{P}}_j$ gates applying a Pauli twirl to that error.}
    \label{fig:jVanishwithCPTP}
\end{figure*}
Then if the $\hat{\mathcal{P}}_j$ gates can successfully twirl the error present, as assumed, to stochastic error, this reduces to as in Fig.~\ref{fig:PostTwirledError}.
\begin{figure*}
\centering
\includegraphics[]{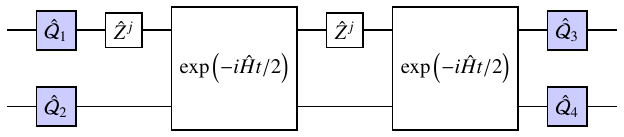}
    \caption{A $j$-vanishing block having experienced CPTP error that has been Pauli twirled to stochastic Pauli error. The resulting stochastic Pauli error is depicted in the above figure by $\hat{\mathcal{Q}}_1$, $\hat{\mathcal{Q}}_2$, $\hat{\mathcal{Q}}_3$, $\hat{\mathcal{Q}}_4 \in \mathbb{P}_1$.}
    \label{fig:PostTwirledError}
\end{figure*}
\end{proof}
\begin{Note}
Lemma~\ref{StochErrorInVanishing} additionally holds if the twirl requires stochastic changes in the two-qubit gates of a vanishing block, provided $\mathbf{N3}$ holds.
\end{Note}
\begin{corollary}
If single-qubit gates experience gate-independent error, 1-vanishing blocks and 0-vanishing blocks experience identical error.
\end{corollary}

\subsubsection{Traps and Target}

The traps and target required for the QAP for circuits with XY-interaction gates are constructed as defined in Algorithms \ref{bandStructureAlgorithm}, \ref{targGenAlgorithm}, and \ref{TrapGenAlgorithm}, in accordance with Algorithm~\ref{StandardAccAlg}.
Informally, the traps and target are constructed by replacing every instance of a two-qubit gate with either a $1$-vanishing block (in the case of traps, which implement the identity) or a $0$-vanishing block (in the case of targets, which implement the replaced gate) -- in the case of traps, the single-qubit gates are also removed before this substitution is made.
These mappings of a circuit to the trap and target circuits are shown in Fig.~\ref{fig:mappingToTarg} and Fig.~\ref{fig:mappingToTrap}.

\begin{figure*}
    \centering
    \includegraphics[]{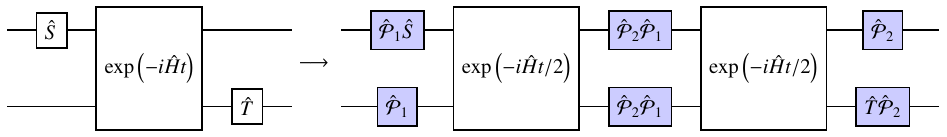}
    \caption{Example of how a subcircuit in the input circuit is mapped to the corresponding subcircuit in the target circuit.}
    \label{fig:mappingToTarg}
\end{figure*}

\begin{figure*}
    \centering
    \includegraphics[]{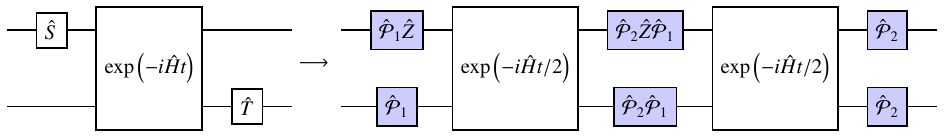}
    \caption{Example of how a subcircuit in the input circuit is mapped to the corresponding subcircuit in the trap circuit.}
    \label{fig:mappingToTrap}
\end{figure*}

These two kinds of vanishing blocks differ only in their single-qubit gates (as can be seen in Fig.~\ref{fig:mappingToTarg} and Fig.~\ref{fig:mappingToTrap}) and so experience identical error, due to $\mathbf{N2}$. In the case of traps, the entire circuit is equivalent (in the error-free case) to the identity, allowing the error to be detected and measured.
\begin{figure*}
    \centering
\begin{algorithm}[H]
%\SetAlgoLined
$\mathbf{Input:}$\\
$\bullet$ An arbitrary circuit, $\mathcal{C}$\\
\vspace{0.1cm}
 \begin{enumerate}
    \item Re-express each single-qubit state preparation in $\mathcal{C}$ as preparing the state $\vert 0 \rangle$ then applying a single-qubit gate
    \item Re-express each single-qubit measurement in $\mathcal{C}$ as applying a single-qubit gate then measuring in the computational basis
\end{enumerate}
\vspace{0.1cm}
$\mathbf{Return}:$ Updated $\mathcal{C}$
\caption{Algorithm for rearranging a circuit into standard form
 \label{bandStructureAlgorithm}}
\end{algorithm}
%\end{figure*}

% Define Targets Algorithm
%\begin{figure*}
    %\centering
\begin{algorithm}[H]
%\SetAlgoLined
$\mathbf{Input:}$ \\
$\bullet$ A circuit, $\mathcal{C}$, where all two-qubit gates are of the form $\exp \left ( -i \hat{H} t \right)$ for some $t \in \mathbb{R}$
\vspace{0.1cm}
 \begin{enumerate}
    \item Convert $\mathcal{C}$ to standard form via Algorithm \ref{bandStructureAlgorithm}
    \item Replace every two-qubit gate in $\mathcal{C}$ with a $0$-vanishing gate
\end{enumerate}
\vspace{0.1cm}
$\mathbf{Return}:$ Updated $\mathcal{C}$
\caption{Algorithm to generate target circuits  
 \label{targGenAlgorithm}}
\end{algorithm}
%\end{figure*}

% Define Traps algorithm
%\begin{figure*}
    %\centering
\begin{algorithm}[H]
%\SetAlgoLined
$\mathbf{Input:}$ \\
$\bullet$ A circuit, $\mathcal{C}$, where all two-qubit gates are of the form $\exp \left ( -i \hat{H} t \right)$ for some $t \in \mathbb{R}$
\vspace{0.1cm}
 \begin{enumerate}
    \item Convert $\mathcal{C}$ to standard form via Algorithm \ref{bandStructureAlgorithm}
    \item Remove all single-qubit gates
    \item Replace every two-qubit gate in $\mathcal{C}$ with a $1$-vanishing gate
    \item With probability $0.5$, place a Hadamard gate on every qubit immediately after state prep and before measurement
    \item For each qubit, with probability $0.5$ add a Pauli Z gate immediately before measurement
    \item For each qubit, with probability $0.5$ add a Pauli Z gate immediately after state preparation
    \item Merge any neighboring Hadamard and Pauli Z gates added in the preceding steps
\end{enumerate}
\vspace{0.1cm}
$\mathbf{Return}:$ Updated $\mathcal{C}$
\caption{Algorithm to generate trap circuits 
 \label{TrapGenAlgorithm}}
\end{algorithm}
\end{figure*}

\begin{theorem}
    \label{XYAccCentralTheorem}
    Circuits subject to errors conforming to N1-N3, where all two-qubit gates are XY-interaction gates can be accredited using Algorithm~\ref{StandardAccAlg}, with the target and trap circuits built using Algorithms~\ref{targGenAlgorithm} and~\ref{TrapGenAlgorithm} respectively.
\end{theorem}
\begin{proof}
We proceed via Theorem~\ref{definitionWorksTheorem}, showing the conditions required for it to apply are met. This then implies that Algorithm~\ref{StandardAccAlg} along with the traps and target circuits built using Algorithms~\ref{TrapGenAlgorithm}~and~\ref{targGenAlgorithm} provide a correctly functioning accreditation protocol. \\

\noindent \underline{Conditions on the target circuit}
\begin{enumerate}
\item [$\mathbf{1a:}$] The target circuit can be created in time linear in the number of gates in the input circuit -- by Algorithm~\ref{targGenAlgorithm}.
\item [$\mathbf{1b:}$] The number of gates in the target can be seen -- in Algorithm~\ref{targGenAlgorithm} -- to be linear in those of the input circuit.
\item [$\mathbf{1c:}$] The number of qubits in the target circuit is exactly the same as in the input circuit: Algorithm~\ref{targGenAlgorithm} adds no qubits.
\item [$\mathbf{1d:}$] Algorithm~\ref{targGenAlgorithm} produces a circuit with exactly the same probabilities of each measurement outcome as in the input circuit.
\item [$\mathbf{1e:}$] As in the original accreditation protocol, the twirling in Algorithm~\ref{targGenAlgorithm} reduces all error in the target to stochastic error -- as in Lemma~\ref{StochErrorInVanishing}.
\end{enumerate}

\noindent \underline{Conditions on the trap circuits}
\begin{enumerate} 
\item [$\mathbf{2a:}$]  The required trap circuits can be created in time linear in the number of gates in the input circuit -- by Algorithm~\ref{TrapGenAlgorithm}.
\item [$\mathbf{2b:}$] The number of gates in any trap can be seen -- in Algorithm~\ref{TrapGenAlgorithm} -- to be linear in those of the input circuit.
\item [$\mathbf{2c:}$] The number of qubits in the trap circuits is exactly the same as in the input circuit: Algorithm~\ref{TrapGenAlgorithm} adds no qubits.
\item [$\mathbf{2d:}$] Lemma~\ref{noErrorTrap} shows the output of any trap circuit not experiencing error is constant and classically computable.
\item [$\mathbf{2e-f:}$] Lemma~\ref{XYStochasticInAcc} implies that all error is equivalent to stochastic Pauli error  and Lemma~\ref{XYStochasticDetectLemma} implies all such error is detected with probability at least $0.5$, as in Algorithm~\ref{TrapGenAlgorithm} the layers of Hadamard gates are present with probability $0.5$.
\end{enumerate}
\noindent \underline{Relation between trap and target error}
\begin{enumerate}
    \item [$\mathbf{3:}$] As trap and target circuits only differ in their single-qubit gates, $\mathbf{N2}$ in the error model implies that the error in trap and target circuits is identical.
\end{enumerate}
\end{proof}

The proof of Theorem~\ref{XYAccCentralTheorem} has used Lemmas~\ref{noErrorTrap},~\ref{XYStochasticInAcc}, and~\ref{XYStochasticDetectLemma}, all proven in Appendix~\ref{sec:ProofOfTheo5Lemmas}.

\section{Robustness of QAP to a Violation of the Assumptions}
%\label{app:WeaklyPerturb}
\subsection{Defining Robustness}
\begin{definition}
    \label{def:robustness}
    If $\mathbb{F}$ is a set of continuous functions, a protocol that functions by executing circuits subject to CPTP error and -- potentially via classical post-computation -- returns a set, $\big \{ V_j \big \}^N_{j = 1}$, of $N$ real values is \underline{$\mathbb{F}$-robust} if and only if replacing the CPTP affecting a specific error operator (or SPAM), $\mathcal{E}$, with a different CPTP map $\mathcal{E}'$, results in a change in $\big \{ V_j \big \}^N_{j = 1}$ such that $\forall j \in \mathbb{N}^{\leq N}$,
    \begin{align}
        \big \vert V_j - V'_j \big \vert \leq f \left( \big \vert \big \vert  \mathcal{E} -\mathcal{E}' \big \vert \big \vert_{\diamond} \right),
    \end{align}
    where $\big \{ \hat{V}'_j \big \}^N_{j = 1}$ are the changed values of $\big \{ V_j \big \}^N_{j = 1}$ (due to $\mathcal{E}'$ replacing $\mathcal{E}$), and $f$ is a  function in the set, $\mathbb{F}$.

    Perhaps the most important of the possible sets $\mathbb{F}$ could be is the set of linear functions, in this case, we call the protocol \underline{linearly robust}.
\end{definition}

\subsection{The Effect of `Weakly' Different Error on Circuit Outputs}
\label{sec:WeaklyPerturb}
The crux of examining how weakly different CPTP errors (i.e., the effect of the dependence is very small, in terms of the diamond norm, compared to overall error) will affect circuit executions is Theorem~\ref{totalBoundRobust}.
\begin{theorem}
    \label{totalBoundRobust}
    For any sub-multiplicative norm $\big \vert \big \vert \cdot  \big \vert \big \vert_{\mathrm{sub}}$, such that for any CPTP map, $\xi$, $\big \vert \big \vert \xi  \big \vert \big \vert_{\mathrm{sub}} \leq 1$ and any circuit, $\mathcal{C}$, comprised of $m \in \mathbb{N}$ gates; if $\mathcal{C}$ has two distinct executions, $\Tilde{\mathcal{C}}^{(1)}$ and $\Tilde{\mathcal{C}}^{(2)}$, each experiencing different CPTP error -- $\big\{ \mathcal{E}^{(1)}_j \big\}_{j = 1}^{m}$ and $\big\{ \mathcal{E}^{(2)}_j \big\}_{j = 1}^{m}$, respectively -- such that there exists a $ \epsilon \in \mathbb{R}$ so that $\forall j \in \mathbb{N}^{\leq m}$, $\big \vert \big \vert \mathcal{E}_{j}^{(1)}
        - \mathcal{E}_{j}^{(2)} \big \vert \big \vert_{\mathrm{sub}} \leq \epsilon$, then
        \begin{align}
            \big \vert \big \vert \Tilde{\mathcal{C}}^{(1)} - \Tilde{\mathcal{C}}^{(2)} \big \vert \big \vert_{\mathrm{sub}}
            \leq
            m \epsilon.
        \end{align}
    \end{theorem}
  Thus, all circuits are linearly robust.
\begin{proof}
Let, $\hat{\mathcal{C}}_j$ be the $j$th gate in the circuit $\mathcal{C}$; $\mathcal{F}_{j}$ be the erroneous execution of $\hat{\mathcal{C}}_j$ in $\Tilde{C}^{(1)}$ (i.e., $\mathcal{F}_{j} = \mathcal{E}^{(1)}_j \hat{\mathcal{C}}_j$); and
$\mathcal{G}_{j}$ be the erroneous execution of $\hat{\mathcal{C}}_j$ in $\Tilde{C}^{(2)}$, i.e., $\mathcal{G}_{j} = \mathcal{E}^{(2)}_j \hat{\mathcal{C}}_j$).

Additionally, define $\mathcal{F}_{k:1}$ by: $\forall k \in \mathbb{N}^{\leq m}$
\begin{align}
    \mathcal{F}_{k:1}
    &=
    \prod_{j = 1}^k \mathcal{F}_j.
\end{align}
$\mathcal{F}_{k:1}$ is thus the composition of all $\mathcal{F}_j$ where $j \leq k$, in order of their indices. $\mathcal{G}_{k:1}$ is defined similarly. 
Then, using the chaining property~\cite[Lemma 9]{jackson2025improvedaccreditationanaloguequantum} of the norm:
    \begin{widetext}
    \begin{align}
    \label{eqn:firstRobustnessBound}
        \big \vert \big \vert \mathcal{F}_{j} -\mathcal{G}_{j} \big \vert \big \vert_{\mathrm{sub}}
        &=
        \big \vert \big \vert \mathcal{E}_{j}^{(1)} \hat{\mathcal{C}}_j
        - \mathcal{E}_{j}^{(2)} \hat{\mathcal{C}}_j \big \vert \big \vert_{\mathrm{sub}}
        =
        \big \vert \big \vert \left( \mathcal{E}_{j}^{(1)}
        - \mathcal{E}_{j}^{(2)} \right)\hat{\mathcal{C}}_j \big \vert \big \vert_{\mathrm{sub}}
        \leq
        \big \vert \big \vert \mathcal{E}_{j}^{(1)}
        - \mathcal{E}_{j}^{(2)} \big \vert \big \vert_{\mathrm{sub}} \big \vert \big \vert \hat{\mathcal{C}}_j \big \vert \big \vert_{\mathrm{sub}}
        \leq
        \big \vert \big \vert \mathcal{E}_{j}^{(1)}
        - \mathcal{E}_{j}^{(2)} \big \vert \big \vert_{\mathrm{sub}}.
        \end{align}
        Therefore, using Eqn.~\eqref{eqn:firstRobustnessBound} to bound the difference between sequences of erroneously implemented gates:
        \begin{align}
        \big \vert \big \vert \mathcal{F}_{m:1} - \mathcal{G}_{m:1} \big \vert \big \vert_{\mathrm{sub}}
        &\leq
        \label{FGSecondEquation}
        \sum_{j = 1}^m \left( \big \vert \big \vert \mathcal{F}_{j} - \mathcal{G}_{j} \big \vert \big \vert_{\mathrm{sub}} \right)
        \leq
        \sum_{j = 1}^m \left( \big \vert \big \vert \mathcal{E}_{j}^{(1)} - \mathcal{E}_{j}^{(2)} \big \vert \big \vert_{\mathrm{sub}} \right).
    \end{align}
   Via the triangle inequality, Eqn.~\eqref{FGSecondEquation}, and the assumed bound, $\big \vert \big \vert \mathcal{E}_{j}^{(1)}
        - \mathcal{E}_{j}^{(2)} \big \vert \big \vert_{\mathrm{sub}} \leq \epsilon$; the difference between the two erroneously implemented circuits can be bounded as:
   \begin{align}
       \big \vert \big \vert \Tilde{\mathcal{C}}^{(1)} - \Tilde{\mathcal{C}}^{(2)} \big \vert \big \vert_{\mathrm{sub}}
       &=
        \bigg \vert \bigg \vert \sum_{\bar{\mathcal{U}}_m} \left( \mathbb{P} \left( \bar{\mathcal{U}}_m \right)  \mathcal{F}_{m:1} \right) 
        -
         \sum_{\bar{\mathcal{U}}_m} \left( \mathbb{P} \left( \bar{\mathcal{U}}_m \right)  \mathcal{G}_{m:1} \right)
         \bigg \vert \bigg \vert_{\mathrm{sub}}
         \leq
         \sum_{\bar{\mathcal{U}}_m} \left( \mathbb{P} \left( \bar{\mathcal{U}}_m \right) \big \vert \big \vert \mathcal{F}_{m:1} - \mathcal{G}_{m:1} \big \vert \big \vert_{\mathrm{sub}} \right) \nonumber \\
         &\leq
         \label{FirstWithE-E}
        \sum_{\bar{\mathcal{U}}_m} \left( \mathbb{P} \left( \bar{\mathcal{U}}_m \right) \sum_{j = 1}^m \left( \big \vert \big \vert \mathcal{E}_{j}^{(1)}
        - \mathcal{E}_{j}^{(2)} \big \vert \big \vert_{\mathrm{sub}} \right ) \right)
    \leq 
    \sum_{\bar{\mathcal{U}}_m} \left( \mathbb{P} \left( \bar{\mathcal{U}}_m \right) \sum_{j = 1}^m \left( \epsilon \right ) \right)
    =
    \sum_{\bar{\mathcal{U}}_m} \left( \mathbb{P} \left( \bar{\mathcal{U}}_m \right) m \epsilon \right)
    =
    m \epsilon.
\end{align}
\end{widetext}
This meets the requirements of Def.~\ref{def:robustness} as the diamond norm is one such norm this proof applies to and it bounds the total variation distance between the two distributions. Thus, the effect on the measurement outcomes is bounded too.
\end{proof}
    Examples of norms that Theorem~\ref{totalBoundRobust} applies to are the diamond norm and the trace norm (nuclear norm).
 
    The results of Theorem~\ref{totalBoundRobust} can be applied when state preparation and measurement are also erroneous. This relies on a trick wherein fictitious identity gates are added immediately after state preparation and before measurement. These identity gates have no effect if they experience no error but if we take the CPTP maps that would model error in the state preparation and measurement and use them to model error in the identity gates we just added, is mathematically identical to error occurring in state preparation and measurement but allows Theorem~\ref{sec:WeaklyPerturb} to account for the state preparation and measurement error. 

We can therefore see, by considering Def.~\ref{def:robustness}, that QAPs are linearly robust. The practical use of this then follows in Corollary~\ref{neglectingDependenceCorrolary}.
\begin{corollary}
    \label{neglectingDependenceCorrolary}
    Consider a circuit execution experiencing error `weakly' dependent on some feature of the execution or circuit (e.g. gate-dependent error). Let $\left \{ \mathcal{E}_{j}^{(1)} \right \}_{j = 1}^m$ represent the true error experienced by the circuit execution and $\left \{ \mathcal{E}_{j}^{(2)} \right \}_{j = 1}^m$ represent some error model independent of the chosen feature of the execution or circuit such that $\forall j \in \mathbb{N}^{\leq m}$, $ \big \vert \big \vert \mathcal{E}_{j}^{(1)}
        - \mathcal{E}_{j}^{(2)} \big \vert \big \vert_{\mathrm{\diamond}} \ll B/ m$. Then the effect of neglecting that feature of the execution is much less than the bound returned by any QAP.

        Informally, if the effect of the weak dependence is small enough, the consequence of neglecting that dependence is negligible.
\end{corollary}
 
\section{Discussion}
\label{discussionSec}
Our key result is a family of Quantum Accreditation Protocols (QAPs) that can be applied in various scenarios.
It is a recipe to easily generate accreditation protocols given of a target computation, a trap computation, and an error model. 
We also prove their correctness.
We applied it to a non-Clifford gate which has a $\hat{\tau}$-decomposition (as in Eqn.\ref{GDef}).
Any circuit where it features as the only two-qubit gate can be accredited.
Our results on the robustness of our protocols to small perturbations and generalization of Pauli twirling to non-Pauli single-qubit bases may be of independent interest.

We also applied our QAP to another set of gates; those with XY-decompositions, at the cost of twirling more limited error. This opens questions around if and when such trade-offs are worthwhile, that are perhaps better addressed by experimentalists, but would allow the error in a wide range of currently-in-use practical gates to be twirled to stochastic Pauli error. This issue is partially resolved with our protocol for twirling XY-decomposible gates, using the extra assumption, but practical use of this depends on how well this assumption holds: a claim best investigated in future work, in consultation with experimentalists.
It is worth noting this second QAP will require two-qubit gates be used as part of twirling but only those that have identical error due to $\mathbf{N3}$. These two-qubit gates can be erroneous. As the error is likely be greater from two-qubit gates, this will likely result in a higher bound being obtained by the QAP.

An evident path forward for the accreditation protocols developed herein is towards experimental implementation of accreditation with non-Clifford native gates. Prominent among them are the $\sqrt{\text{iSWAP}}$ and $\text{fSWAP}.$
Our numerical searches show that these do not allow a $\hat{\tau}$-decomposition.

We consider the work on $\hat{\tau}$-decomposible gates more of an illustration that non-Clifford gates can be twirled and a potential first step towards twirling larger sets of two-qubit gates.
Further development of the techniques originated in this paper to accredit circuits with non-Clifford gates spanning two or more qubits are left for future work.

\section{Acknowledgements}
This work was supported, in part, 
by an EPSRC IAA grant (G.PXAD.0702.EXP), the UKRI ExCALIBUR project QEVEC (EP/W00772X/2), the Quantum Advantage Pathfinder (EP/X026167/1), the Hub for Quantum Computing via Integrated and Interconnected Implementations (QCI3) (EP/Z53318X/1).

\newpage
\bibliography{main_v3}

%%%%%%%%%%%%%%%%%%%%%%%%%%%% APPENDIX %%%%%%%%%%%%%%%%%%%%%
\newpage
\onecolumngrid
\appendix
\section{Formal Definitions of Twirlable and Accreditable Error}
\label{appendix:defsOFTwirlableAndAccreditable}

\begin{definition}
    Relative to a given two-qubit gate, a family of CPTP maps are termed \underline{twirlable} if the error they represent, when it occurs in that two-qubit gate, can be effectively reduced to stochastic error.
\end{definition}

\begin{definition}
    Relative to a given two-qubit gate, a family of CPTP maps are termed \underline{accreditable} if, when the family of errors they represent is the only family of errors that can occur in implementations of the given two-qubit gate, there exists an accreditation protocol that functions correctly.
\end{definition}

\section{Proof of Lemma~\ref{GgateCommuteLemma}}
\label{appendixProofofLemmaGgateCommuteLemma}
\begin{definition}
    \label{def:postCommutePaulis}
    For any Clifford gate, $\hat{\mathcal{M}}$ and single-qubit Pauli gate, $\hat{P}_j$, define the operators \underline{$\Bar{P}^{(1)}_{j, \hat{\mathcal{M}}}$} and \underline{$\Bar{P}^{(2)}_{j, \hat{\mathcal{M}}}$} by:
    \begin{align}
        \label{eqn:p1def}
        \Bar{P}^{(1)}_{j, \hat{\mathcal{M}}} &= \hat{\mathcal{M}}^{\dagger} \circ \left( \hat{P}_j \otimes \hat{I} \right) \circ \hat{\mathcal{M}},\\
        \label{eqn:p2def}
       \text{ and } \Bar{P}^{(2)}_{j, \hat{\mathcal{M}}} &=
        \hat{\mathcal{M}}^{\dagger} \circ \left( \hat{I} \otimes \hat{P}_j \right) \circ \hat{\mathcal{M}}.
    \end{align}
    Eqn.~\eqref{eqn:p1def} is equivalent to the circuit equation in Fig.~\ref{fig:PbarDefGraphic}. Eqn.~\eqref{eqn:p2def} has an equivalent circuit equation: Fig.~\ref{fig:PbarDefGraphic}.
    \begin{figure}[h]
    \includegraphics[]{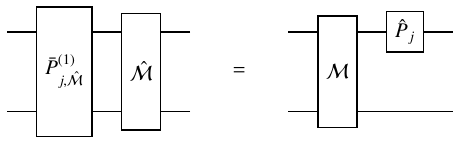}
\caption{A graphical depiction of Eqn.~\eqref{eqn:p1def}.  This is the key property of $\Bar{P}^{(1)}_{j, \hat{\mathcal{M}}}$ for twirling error. $\Bar{P}^{(1)}_{j, \hat{\mathcal{M}}}$ consists of two single-qubit gates.}
\label{fig:PbarDefGraphic}
\end{figure}

    As $\mathcal{M}$ is Clifford, $\Bar{P}^{(1)}_{j, \hat{\mathcal{M}}}$ and $\Bar{P}^{(2)}_{j, \hat{\mathcal{M}}}$ will always -- up to an overall phase -- be in $\mathbb{P}_1^{\otimes 2}$.
    
    For example, if $\hat{\mathcal{M}}$ is a cNOT gate (with the control on the first qubit) and $\hat{P}_j=\hat{X}$, then $\Bar{P}^{(1)}_{j, \hat{\mathcal{M}}} = \hat{X} \otimes \hat{X}$ and $\Bar{P}^{(2)}_{j, \hat{\mathcal{M}}} = \hat{I} \otimes \hat{X}$.
\end{definition}

\begin{proof}
Starting from the left-hand-side of Eqn.~\eqref{PprimeQprimeMainProperty} and using the $\hat{\tau}$-decomposition of $\hat{\mathcal{G}}$ (as in Def.~\ref{def:GammaG}):
\begin{align}
    \left(  \hat{\tau}_1^{\dagger} \hat{P}_j \hat{\tau}_1 \otimes \hat{\tau}_2^{\dagger} \hat{P}_k \hat{\tau}_2 \right) \circ \hat{\mathcal{G}}
    &=
    \label{eqn:withTauDecomp}
    \left(  \hat{\tau}_1^{\dagger} \hat{P}_j \hat{\tau}_1 \otimes \hat{\tau}_2^{\dagger} \hat{P}_k \hat{\tau}_2 \right) \circ \left( \hat{\tau}_1 \otimes \hat{\tau}_2 \right)^{\dagger} \circ \hat{\mathcal{M}} \circ \left( \hat{\tau}_1 \otimes \hat{\tau}_2 \right).
\end{align}
    The cancellation of $\left( \hat{\tau}_1 \otimes \hat{\tau}_2 \right)$ with its Hermitian conjugate, before we push the Pauli gates ($\hat{P}_j$ and $\hat{P}_k$) through $\hat{\mathcal{M}}$, allows Eqn.~\eqref{eqn:withTauDecomp} to be rewritten -- using Def.~\ref{def:postCommutePaulis} -- as:
    \begin{align}
    \left(  \hat{\tau}_1^{\dagger} \hat{P}_j \hat{\tau}_1 \otimes \hat{\tau}_2^{\dagger} \hat{P}_k \hat{\tau}_2 \right) \circ \hat{\mathcal{G}}
    =
    \left( \hat{\tau}_1\otimes \hat{\tau}_2 \right)^{\dagger} \circ \left( \hat{P}_j \otimes \hat{P}_k \right) \circ \hat{\mathcal{M}} \circ \left( \hat{\tau}_1 \otimes \hat{\tau}_2 \right)
    =
    \label{eqn:Pre-Sep_qubits}
    \left( \hat{\tau}_1\otimes \hat{\tau}_2 \right)^{\dagger} \circ \hat{\mathcal{M}} \circ \Bar{P}^{(1)}_{j, \hat{\mathcal{M}}} \circ \Bar{P}^{(2)}_{k, \hat{\mathcal{M}}} \circ \left( \hat{\tau}_1 \otimes \hat{\tau}_2 \right).
\end{align}
Let $\textrm{Tr}_{q} \left [ \cdot \right ]$ denote the partial trace~\cite{BHATIA2003125}  over qubit $q$ of its argument\footnote{Throughout the rest of this proof tracing out one qubit will always result in just a single-qubit gate acting on the other.}
\emph{For example}, if $\Bar{P}^{(1)}_{j, \hat{\mathcal{M}}} = \hat{X} \otimes \hat{Z}$ and $\Bar{P}^{(2)}_{j, \hat{\mathcal{M}}} = \hat{Y} \otimes \hat{I}$ then:
\begin{align} \textit{Tr}_{2} \bigg[ \Bar{P}^{(1)}_{j, \hat{\mathcal{M}}} \circ \Bar{P}^{(2)}_{k, \hat{\mathcal{M}}} \bigg] = \hat{Z} \text{ and } \textit{Tr}_{1} \bigg[ \Bar{P}^{(1)}_{j, \hat{\mathcal{M}}} \circ \Bar{P}^{(2)}_{k, \hat{\mathcal{M}}} \bigg] = \hat{Z},
\end{align}
up to an overall phase.

Returning to the main argument, using the partial traces and that $\Bar{P}^{(1)}_{j, \hat{\mathcal{M}}} \circ \Bar{P}^{(2)}_{k, \hat{\mathcal{M}}}$ will always be a tensor product of two single-qubit Pauli matrices so either partial trace of it will always be a single single-qubit Pauli matrix, the right-hand-side of Eqn.~\eqref{eqn:Pre-Sep_qubits} can be expressed as the tensor product of an element of $\Gamma_{\hat{\tau}_1}$ with an element of $\Gamma_{\hat{\tau}_2}$ post-composed with $\hat{\mathcal{G}}$, i.e.,
\begin{align}
     \left(  \hat{\tau}_1^{\dagger} \hat{P}_j \hat{\tau}_1 \otimes \hat{\tau}_2^{\dagger} \hat{P}_k \hat{\tau}_2 \right) \circ \hat{\mathcal{G}}
    =
    \label{prePprimeQprime}
    \hat{\mathcal{G}} \circ \left( \bigg \{ \hat{\tau}_1^{\dagger} \circ \textit{Tr}_{2} \bigg[ \Bar{P}^{(1)}_{j, \hat{\mathcal{M}}} \circ \Bar{P}^{(2)}_{k, \hat{\mathcal{M}}} \bigg] \circ \hat{\tau}_1 \bigg\} \otimes 
\bigg\{ \hat{\tau}_2^{\dagger} \circ \textit{Tr}_{1} \bigg[ \Bar{P}^{(1)}_{j, \hat{\mathcal{M}}} \circ \Bar{P}^{(2)}_{k, \hat{\mathcal{M}}} \bigg] \circ \hat{\tau}_2 \bigg \} \right)
=
\hat{\mathcal{G}} \circ \left( \hat{P}^{\prime}_{j,k} \otimes \hat{Q}^{\prime}_{j,k} \right),
\end{align}
where:
\begin{align}
    \label{PPrimeDef}
    \hat{P}^{\prime}_{j,k} 
    &=
    \hat{\tau}_1^{\dagger} \circ \textit{Tr}_{2} \bigg[ \Bar{P}^{(1)}_{j, \hat{\mathcal{M}}} \circ \Bar{P}^{(2)}_{k, \hat{\mathcal{M}}} \bigg] \circ \hat{\tau}_1\\
    \label{QPrimeDef}
    \text{ and } \hat{Q}^{\prime}_{j,k} 
    &=
    \hat{\tau}_2^{\dagger} \circ \textit{Tr}_{1} \bigg[ \Bar{P}^{(1)}_{j, \hat{\mathcal{M}}} \circ \Bar{P}^{(2)}_{k, \hat{\mathcal{M}}} \bigg] \circ \hat{\tau}_2.
\end{align}
As $ \hat{\tau}_1^{\dagger} \circ \textit{Tr}_{2} \bigg[ \Bar{P}^{(1)}_{j, \hat{\mathcal{M}}} \circ \Bar{P}^{(2)}_{k, \hat{\mathcal{M}}} \bigg] \circ \hat{\tau}_1  \in \Gamma_{\hat{\tau}_1}$ and $ \hat{\tau}_2^{\dagger} \circ \textit{Tr}_{1} \bigg[ \Bar{P}^{(1)}_{j, \hat{\mathcal{M}}} \circ \Bar{P}^{(2)}_{k, \hat{\mathcal{M}}} \bigg] \circ \hat{\tau}_2\in \Gamma_{\hat{\tau}_2}$ and both are single-qubit gates, Lemma~\ref{GgateCommuteLemma} can be seen to be true, with $\hat{P}^{\prime}_{j,k} \in \Gamma_{\hat{\tau}_1}$ and $\hat{Q}^{\prime}_{j,k} \in \Gamma_{\hat{\tau}_2}$ (as used in Eqn.~\eqref{prePprimeQprime}).
\end{proof}
\begin{Note}
    Additionally, for any $\hat{P}^{\prime}_{j,k} \in \Gamma_{\hat{\tau}_1}$ and $\hat{Q}^{\prime}_{j,k} \in \Gamma_{\hat{\tau}_2}$, there exists a $\left(  \hat{\tau}_1^{\dagger} \hat{P}_j \hat{\tau}_1 \otimes \hat{\tau}_2^{\dagger} \hat{P}_k \hat{\tau}_2 \right) \in \Gamma_{\hat{\mathcal{G}}}$, such that Eqn.~\eqref{PprimeQprimeMainProperty} holds.
\end{Note}

\section{Twirling in a new Basis}
\label{TwirlingTauDeets}
\subsection{Pre-requisites: Commutator Functions and Superoperators}
Before beginning the proof of Lemma~\ref{twirlLemma}, we require two definitions of objects featured in the proof of Lemma  \ref{twirlLemma}.\\
The first of these is commutator functions.
\begin{definition}
\label{commutatorsProperDef}
    Define the \underline{commutator function}, $\xi$, between two operators ($\hat{a}$ and $\hat{b}$) by: $ \hat{a} \circ \hat{b} = \xi(\hat{a},\hat{b}) \hat{b} \circ \hat{a}$.\\
    $\forall N \in \mathbb{N}$, if both arguments of the commutator function are in $\mathbb{P}^{\otimes N}$ then it always always returns a value in $\{ 1, -1\}$, as all elements of $\mathbb{P}^{\otimes N}$ either commute or anti-commute with each other.
\end{definition}
These commutator functions have the important properties (stated without proof, but proofs can be found in Ref. \cite{Cai_2019}) that:
\begin{lemma}
\label{xiProductLemma}
    $\forall n \in \mathbb{N}$, $\forall g, g_1, g_2 \in \mathbb{U}\left(n \right)$, $
        \xi \left( g, g_1 \right) \xi \left( g, g_2 \right)
        =
        \xi \left( g, g_1g_2\right)$.
\end{lemma}

\begin{lemma}
    \label{lem:reverseOrderXi}
    $\forall N \in \mathbb{N}$, $\forall a, b \in \mathbb{P}_1^{\otimes N}$,
        $\xi \left( a, b \right)
        =
        \xi \left( b, a\right)$.
\end{lemma}

\begin{lemma}
    For any $\hat{\tau}_1, \hat{\tau}_2 \in \mathbb{U}(2)$, there exists a bijection, $f: \Gamma_{\hat{\tau}_1} \otimes \Gamma_{\hat{\tau}_2} \rightarrow \mathbb{P}_1^{\otimes 2}$ such that, $\forall \hat{\gamma}_1, \hat{\gamma}_2 \in \Gamma_{\hat{\tau}_1} \otimes \Gamma_{\hat{\tau}_2}$, 
    \begin{align}
    \label{eqn:xiDef}
    \xi \left( f (\hat{\gamma}_1), f(\hat{\gamma}_2) \right) 
    =
    \xi \left( \hat{\gamma}_1, \hat{\gamma}_2 \right).
    \end{align}
\end{lemma}
\begin{proof}
    We propose the solution for Eqn.~\eqref{eqn:xiDef}: $\forall \hat{\gamma} \in \Gamma_{\hat{\tau}_1} \otimes \Gamma_{\hat{\tau}_2}$, $f(\hat{\gamma}) = \left( \hat{\tau}_1 \otimes \hat{\tau}_2 \right)^{\dagger} \hat{\gamma} \left( \hat{\tau}_1 \otimes \hat{\tau}_2 \right)$.

    We then check that this is a correct solution by calculating:
    \begin{align}
        f(\hat{\gamma}_2)  f (\hat{\gamma}_1)
        &=
        \xi \left( f (\hat{\gamma}_1), f(\hat{\gamma}_2) \right) f (\hat{\gamma}_1), f(\hat{\gamma}_2)\\
        \Rightarrow
        \left( \hat{\tau}_1 \otimes \hat{\tau}_2 \right)^{\dagger} \hat{\gamma}_1 \left( \hat{\tau}_1 \otimes \hat{\tau}_2 \right) \left( \hat{\tau}_1 \otimes \hat{\tau}_2 \right)^{\dagger} \hat{\gamma}_2 \left( \hat{\tau}_1 \otimes \hat{\tau}_2 \right)
        &=
        \xi \left( f (\hat{\gamma}_1), f(\hat{\gamma}_2) \right) \left( \hat{\tau}_1 \otimes \hat{\tau}_2 \right)^{\dagger} \hat{\gamma}_1 \left( \hat{\tau}_1 \otimes \hat{\tau}_2 \right) \left( \hat{\tau}_1 \otimes \hat{\tau}_2 \right)^{\dagger} \hat{\gamma}_2 \left( \hat{\tau}_1 \otimes \hat{\tau}_2 \right)\\
        \Rightarrow
        \left( \hat{\tau}_1 \otimes \hat{\tau}_2 \right)^{\dagger} \hat{\gamma}_1 \hat{\gamma}_2 \left( \hat{\tau}_1 \otimes \hat{\tau}_2 \right)
        &=
        \xi \left( f (\hat{\gamma}_1), f(\hat{\gamma}_2) \right) \left( \hat{\tau}_1 \otimes \hat{\tau}_2 \right)^{\dagger} \hat{\gamma}_1 \hat{\gamma}_2 \left( \hat{\tau}_1 \otimes \hat{\tau}_2 \right)\\
        \Rightarrow
        \hat{\gamma}_1 \hat{\gamma}_2
        &=
        \xi \left( f (\hat{\gamma}_1), f(\hat{\gamma}_2) \right) \hat{\gamma}_1 \hat{\gamma}_2.
    \end{align}
    This confirms that $\xi \left( f (\hat{\gamma}_1), f(\hat{\gamma}_2) \right) 
    =
    \xi \left( \hat{\gamma}_1, \hat{\gamma}_2 \right).$
\end{proof}
\begin{corollary}
\label{equivSumCorollary}
    $\forall \hat{\gamma} \in \Gamma_{\hat{\tau}_1} \otimes \Gamma_{\hat{\tau}_2}$, $\exists \hat{\mathcal{P}} \in \mathbb{P}_1^{\otimes 2}$ such that:
    \begin{align}
        \sum_{\hat{\lambda} \in \Gamma} \left( \xi \left( \hat{\lambda}, \hat{\gamma}\right) \right)
        =
        \sum_{\hat{\lambda}' \in \mathbb{P}_1^{\otimes 2}} \left( \xi \left( \hat{\lambda}', \hat{\mathcal{P}} \right) \right).
    \end{align}
\end{corollary}

The other objects to be defined are super-operators, in Def.~\ref{def:superoperators}.
\begin{definition}
    \label{def:superoperators}
    For any operator, $A$, define the corresponding \underline{super-operator}, $\overbracket{A}$, by, for any density matrix, $\rho$:
$
    \overbracket{A} \rho 
    =
    A \rho A^{\dagger}$
\end{definition}
\subsection{Core Algebraic Mechanics of Twirling}
\begin{Note}
    The below Lemma \ref{MostGeneralTwirlLemma} is a generalization and abstractification of in the section titled ``Requirements and results of twirling'' 
 in Ref.~\cite{Cai_2019}. It follows roughly the same approach.
\end{Note}
\begin{lemma}
    \label{MostGeneralTwirlLemma}
    For any set of operators, $\mathbb{S}$, acting on a Hilbert space, $H$, if $\exists \Lambda, Q \subseteq \mathbb{S}$ such that:\\
    $\mathbf{S1}:$ $\mathbb{S}$ is closed under composition,\\
    $\mathbf{S2}:$ $\forall \hat{a}, \hat{b} \in \mathbb{S}$, $\xi \left( \hat{a}, \hat{b} \right) = \xi \left( \hat{b}, \hat{a} \right) \propto \hat{I}$,\\
    $\mathbf{Q1}:$ $\forall \hat{a}, \hat{b} \in Q$, $\textit{Tr} \left [ \hat{a} \circ \hat{b} \right ] = \Xi \delta_{\hat{a}, \hat{b}}$, where $ \Xi \in \mathbb{R}$.\\
    and\\
    \underline{The $\Lambda$ Summation Condition:}
    $\forall \hat{\gamma}_1, \hat{\gamma}_2 \in Q$,
        $\sum_{\hat{\lambda} \in \Lambda} \left( \xi \left( \hat{\gamma}_1 \circ \hat{\gamma}_2, \hat{\lambda} \right) \right)
        =
        \delta_{\hat{\gamma}_1, \hat{\gamma}_2} \big \vert \Lambda \big \vert$,
    \\ and\\
    \underline{The $\Lambda$ Self-Adjoint Condition:}
    $\forall \hat{\lambda} \in \Lambda$,
        $\hat{\lambda} = \hat{\lambda}^{\dagger}$,\\
    then conjugation by uniformly random elements from $\Lambda$ effectively reduces any CPTP error, with Kraus operators that all have a decomposition as a linear combination of elements in Q, to stochastic error consisting of elements from $Q$.
\end{lemma}
\begin{proof}
    Assume $\overbracket{\mathcal{E}}$ is a superoperator representing error conforming to all conditions in the lemma statement. Define the conjugation of the error with uniformly random elements of $\Lambda$ by $\mathcal{T}_{\Lambda}^N \left( \overbracket{\mathcal{E}} \right)$. Then, $\mathcal{T}_{\Lambda}^N \left( \overbracket{\mathcal{E}} \right)$ can be expressed as:
    \begin{align}
        \label{eqn:TwirlFirstInProof}
        \mathcal{T}_{\Lambda}^N \left( \overbracket{\mathcal{E}} \right)
        &=
        \dfrac{1}{\big \vert \Lambda \big \vert} \sum_{\hat{\lambda} \in \Lambda} \left( \overbracket{\hat{\lambda} \circ \mathcal{E} \circ \hat{\lambda}^{\dagger}} \right).
    \end{align}
    As $\mathcal{E}$ is a CPTP map, it has a Kraus decomposition. By the assumption that every Kraus operator of $\mathcal{E}$ has a decomposition as a linear combination of elements from $Q$, and using assumption $\mathbf{Q1}$: for any Kraus operator of $\mathcal{E}$, $\hat{\mathcal{K}}$,
    \begin{align}
        \hat{\mathcal{K}}
        &=
        \sum_{\gamma \in Q} \left( \dfrac{\textit{Tr}\left [ \gamma \hat{\mathcal{K}} \right ]}{\Xi} \gamma \right).
    \end{align}
    Therefore, using the Kraus decomposition of $\mathcal{E}$, for any state, $\rho$,
    \begin{align}
        \label{eqn:AppedixInTwirlToSubIn}
        \overbracket{\hat{\lambda} \circ \mathcal{E} \circ \hat{\lambda}^{\dagger}} \rho
        &=
        \left( \dfrac{1}{\Xi} \right)^2
        \sum_{\substack{\hat{\gamma}_1, \hat{\gamma}_2 \in Q\\ \hat{\mathcal{K}}}} \left( \textit{Tr}\left [ \hat{\gamma}_1\hat{\mathcal{K}} \right ] \bigg[ \textit{Tr}\left [ \hat{\gamma}_2 \hat{\mathcal{K}} \right ]\bigg]^{*} \hat{\lambda} \circ \hat{\gamma}_1 \circ \hat{\lambda}^{\dagger} \circ \rho \circ \left [\hat{\lambda} \circ \hat{\gamma}_2 \circ \hat{\lambda}^{\dagger} \right ]^{\dagger} \right),
    \end{align}
    where the $\hat{\mathcal{K}}$ in the subscript of a summation denotes summing over each Kraus operator of $\mathcal{E}$.
    Therefore, substituting Eqn.~\eqref{eqn:AppedixInTwirlToSubIn} into Eqn.~\eqref{eqn:TwirlFirstInProof}, then suppressing each $\circ$ and using the $\Lambda$ self-adjoint condition:
    \begin{align}
         \mathcal{T}_{\Lambda}^N \left( \overbracket{\mathcal{E}} \right) \rho
         &=
         \label{withTwoSumsInGeneralTwirl}
         \dfrac{1}{\big \vert \Lambda \big \vert} \left( \dfrac{1}{\Xi} \right)^2
        \sum_{\substack{\hat{\gamma}_1, \hat{\gamma}_2 \in Q\\ \hat{\mathcal{K}}}} \left( \textit{Tr}\left [ \hat{\gamma}_1\hat{\mathcal{K}} \right ] \bigg[ \textit{Tr}\left [ \hat{\gamma}_2 \hat{\mathcal{K}} \right ]\bigg]^{*} \sum_{\hat{\lambda} \in \Lambda} \left( \hat{\lambda} \hat{\gamma}_1  \hat{\lambda} \rho \hat{\lambda}  \hat{\gamma}_2 \hat{\lambda}  \right) \right).
    \end{align}
    Consider just the inner sum of Eqn.~\eqref{withTwoSumsInGeneralTwirl}. Using condition $\mathbf{S2}$ and the $\Lambda$ self-adjoint condition:
    \begin{align}
        \sum_{\hat{\lambda} \in \Lambda} \left( \hat{\lambda} \hat{\gamma}_1  \hat{\lambda} \rho \hat{\lambda} \hat{\gamma}_2 \hat{\lambda}  \right)
        &=
        \sum_{\hat{\lambda} \in \Lambda} \left( \xi \left( \hat{\gamma}_1, \hat{\lambda} \right) \xi \left( \hat{\gamma}_2, \hat{\lambda} \right) \hat{\gamma}_1 \rho \hat{\gamma}_2  \right)
        =
        \sum_{\hat{\lambda} \in \Lambda} \left( \xi \left( \hat{\gamma}_1 \hat{\gamma}_2, \hat{\lambda} \right)  \right)  \hat{\gamma}_1 \rho \hat{\gamma}_2
        =
        \label{toResub}
        \big \vert \Lambda \big \vert \delta_{\hat{\gamma}_1, \hat{\gamma}_2} \hat{\gamma}_1 \rho \hat{\gamma}_2,
    \end{align}
    where this has also used Lemma~\ref{xiProductLemma} and Lemma~\ref{lem:reverseOrderXi}.
   Substituting Eqn.~\eqref{toResub} into Eqn.~\eqref{withTwoSumsInGeneralTwirl}: for any state, $\rho$,
    \begin{align}
        \label{eqn:FinalTwirlEquation}
         \mathcal{T}_{\Lambda}^N \left( \overbracket{\mathcal{E}} \right) \rho
         &=
         \dfrac{1}{\big \vert \Lambda \big \vert} \left( \dfrac{1}{\Xi} \right)^2
        \sum_{\substack{\hat{\gamma}_1, \hat{\gamma}_2 \in Q\\ \hat{\mathcal{K}}}} \left( \textit{Tr}\left [ \hat{\gamma}_1\hat{\mathcal{K}} \right ] \bigg[ \textit{Tr}\left [ \hat{\gamma}_2 \hat{\mathcal{K}} \right ]\bigg]^{*} \big \vert \Lambda \big \vert \delta_{\hat{\gamma}_1, \hat{\gamma}_2} \hat{\gamma}_1  \rho \hat{\gamma}_2 \right)
        =
        \sum_{\substack{\hat{\gamma}_1\in Q\\ \hat{\mathcal{K}}}} \left( \dfrac{\big \vert \textit{Tr}\left [ \hat{\gamma}_1\hat{\mathcal{K}} \right ] \big \vert^2}{\Xi^2} \hat{\gamma}_1 \rho \hat{\gamma}_1 \right)
        \Rightarrow
        \mathcal{T}_{\Lambda}^N \left( \overbracket{\mathcal{E}} \right)
        =
        \sum_{\substack{\hat{\gamma}_1 \in Q\\ \hat{\mathcal{K}}}} \left( \dfrac{\big \vert \textit{Tr}\left [ \hat{\gamma}_1\hat{\mathcal{K}} \right ] \big \vert^2}{\Xi^2} \overbracket{\hat{\gamma}_1} \right).
    \end{align}
    As, for any $\hat{\gamma}_1 \in \mathbb{U}(2)$ and set of Kraus operators, $\sum_{\hat{\mathcal{K}}} \left( \vert \textit{Tr}\left [ \hat{\gamma}_1 \hat{\mathcal{K}} \right ] \big \vert^2 \right) \in \mathbb{R}$ (and $\Xi \in \mathbb{R}$, due to $\mathbf{Q1}$), the twirled error, $\mathcal{T}_{\Lambda}^N \left( \overbracket{\mathcal{E}} \right)$ (at the end of Eqn.~\eqref{eqn:FinalTwirlEquation}), can be interpreted as stochastic error consisting of elements from $Q$. 
\end{proof}
\subsection{Proof of Lemma~\ref{twirlLemma}}
\begin{lemmaRestate}[Lemma~\ref{twirlLemma}]
If any CPTP map is conjugated by uniformly random elements from $\Gamma_{\hat{\mathcal{G}}}$, it is effectively (in terms of the probabilities of measurement outcomes) reduced to $\Gamma_{\hat{\mathcal{G}}}$-stochastic error.
\end{lemmaRestate}
\begin{proof}[Proof of Lemma~\ref{twirlLemma}]
    This proof consists of showing that $\forall \hat{\tau}_1,\hat{\tau}_2 \in \mathbb{U}(2)$, Lemma~\ref{MostGeneralTwirlLemma} can be applied by setting $\mathbb{S} = Q = \Lambda = \Gamma_{\hat{\tau}_1} \otimes \Gamma_{\hat{\tau}_2}$. This then directly implies Lemma~\ref{twirlLemma}.

    Each condition of Lemma~\ref{MostGeneralTwirlLemma} is addressed in turn.

    $\Gamma_{\hat{\tau}_1} \otimes \Gamma_{\hat{\tau}_2}$ inherits its group structure from the Pauli group on two spins/qubits (as it is isomorphic to this group), so conditions $\mathbf{S1}$ and $\mathbf{S2}$ are met.
    
    This inheritance of group structure from the Pauli group (used above), along with the cyclic invariance of the trace, also implies condition $\mathbf{Q1}$.
    
    Showing that the $\Lambda$ summation condition is satisfied is a little more involved. First consider: $\forall \hat{\gamma}_1, \hat{\gamma}_2 \in \Gamma_{\hat{\tau}_1} \otimes \Gamma_{\hat{\tau}_2}$,
    \begin{align}
        \hat{\gamma}_1 \circ \hat{\gamma}_2 \in \begin{cases}
            \left( \Gamma_{\hat{\tau}_1} \otimes \Gamma_{\hat{\tau}_2} \right) \backslash \{ \hat{I} \}, &\mathrm{ if } \hspace{0.4cm} \hat{\gamma}_1 = \hat{\gamma}_2\\
            \{\hat{I}\}, &\mathrm{ if } \hspace{0.4cm} \hat{\gamma}_1 \not = \hat{\gamma}_2 
        \end{cases}.
    \end{align}
    Therefore, $\forall \hat{\gamma}_1, \hat{\gamma}_2 \in \Gamma_{\hat{\tau}_1} \otimes \Gamma_{\hat{\tau}_2}$, $\exists \hat{\gamma}_3 \in \Gamma_{\hat{\tau}_1} \otimes \Gamma_{\hat{\tau}_2} \backslash \{ \hat{I} \}$ such that:
    \begin{align}
        \sum_{\hat{\lambda} \in \Gamma_{\hat{\tau}_1} \otimes \Gamma_{\hat{\tau}_2}} \left( \xi \left( \hat{\gamma}_1 \circ \hat{\gamma}_2, \hat{\lambda} \right) \right)
        &=
        \begin{cases}
            \big \vert \Gamma \big \vert, &\mathrm{ if } \hspace{0.4cm} \hat{\gamma}_1 = \hat{\gamma}_2\\
             \sum_{\hat{\lambda} \in \Gamma_{\hat{\tau}_1} \otimes \Gamma_{\hat{\tau}_2}} \left( \xi \left( \hat{\gamma}_3, \hat{\lambda} \right) \right), &\mathrm{ if } \hspace{0.4cm} \hat{\gamma}_1 \not = \hat{\gamma}_2 
        \end{cases}.
    \end{align}
    Using Corollary~\ref{equivSumCorollary}, $\forall \hat{\tau}_1, \hat{\tau}_2 \in \mathbb{U}(2)$, $\forall \hat{\gamma}_3 \in \Gamma_{\hat{\tau}_1} \otimes \Gamma_{\hat{\tau}_2}$, $\sum_{\hat{\lambda} \in \Gamma_{\hat{\tau}_1} \otimes \Gamma_{\hat{\tau}_2}} \left( \xi \left( \hat{\gamma}_3, \hat{\lambda} \right) \right)$ equates to $\sum_{\hat{\lambda}' \in \mathbb{P}_1^{\otimes 2}} \left( \xi \left( \hat{\mathcal{P}}, \hat{\lambda}' \right) \right)$ for some $\hat{\mathcal{P}} \in \mathbb{P}_1^{\otimes 2}$. This reduces showing that the $\Lambda$ summation condition is met to checking a finite number of sums, (i.e,  for each possible value of $\hat{\mathcal{P}} \in \mathbb{P}_1^{\otimes 2}$) always equal zero. Therefore, the summation condition is met.\\
    
    Finally we turn to examine the $\Lambda$ self-adjoint condition. As $\forall \hat{\lambda} \in \Gamma_{\hat{\tau}_1} \otimes \Gamma_{\hat{\tau}_2}$, $\exists \hat{\mathcal{P}}, \hat{\mathcal{P}}^{\prime} \in \mathbb{P}_1$ such that $\hat{\lambda} = (\hat{\tau}_1 \otimes \hat{\tau}_2)^{\dagger} (\hat{\mathcal{P}} \otimes \hat{\mathcal{P}}^{\prime}) (\hat{\tau}_1 \otimes \hat{\tau}_2)$, $\forall \hat{\lambda} \in \Gamma_{\hat{\tau}_1} \otimes \Gamma_{\hat{\tau}_2}$,
    \begin{align}
        \hat{\lambda}^{\dagger}
        =
        \left [ (\hat{\tau}_1 \otimes \hat{\tau}_2)^{\dagger} (\hat{\mathcal{P}} \otimes \hat{\mathcal{P}}^{\prime}) (\hat{\tau}_1 \otimes \hat{\tau}_2) \right]^{\dagger}
        =
        (\hat{\tau}_1 \otimes \hat{\tau}_2)^{\dagger} (\hat{\mathcal{P}} \otimes \hat{\mathcal{P}}^{\prime})^{\dagger} \left [ (\hat{\tau}_1 \otimes \hat{\tau}_2)^{\dagger} \right]^{\dagger}
        =
        (\hat{\tau}_1 \otimes \hat{\tau}_2)^{\dagger} (\hat{\mathcal{P}} \otimes \hat{\mathcal{P}}^{\prime}) (\hat{\tau}_1 \otimes \hat{\tau}_2)
        =
        \hat{\lambda}.
    \end{align}
    Hence, the $\Lambda$ self-adjoint condition is also met.
\end{proof}

\subsection{Twirling Error in State Preparation and Measurement}
Error does not just occur in the non-SPAM operators. It also occurs in state preparation and measurement (collectively referred to as SPAM). This error too can be twirled into stochastic error, as shown in Lemmas~\ref{statePrepTwirlLemma} and~\ref{measTwirlError}.

\begin{lemma}
    \label{statePrepTwirlLemma}
    Error in the preparation of the state $\vert \hat{\tau}^{\dagger} \rangle \langle \hat{\tau}^{\dagger} \vert = \hat{\tau}^{\dagger} \vert 0 \rangle \langle 0 \vert \hat{\tau}$ can be twirled into $\Gamma_{\hat{\tau}}$-stochastic error, for any $\hat{\tau} \in \mathbb{U} (2)$, on the same terms as Pauli twirling state preparation.
\end{lemma}
\begin{proof}
    Consider preparing the state $\vert \hat{\tau}^{\dagger} \rangle \langle \hat{\tau}^{\dagger} \vert$. Then, with probability $0.5$ apply the gate $\hat{\tau}^{\dagger} \hat{Z} \hat{\tau}$. If CPTP error, $\mathcal{E}$, occurs after state preparation, this results in the mixed state (which we denote as $J$):
    \begin{align}
        J
        &=
        \dfrac{1}{2} \sum_{k \in \{ 0,1 \} } \left( \left( \hat{\tau}^{\dagger}\hat{Z} \hat{\tau} \right)^k \mathcal{E} \left( \vert \hat{\tau}^{\dagger} \rangle \langle \hat{\tau}^{\dagger} \vert \right) \left( \left [ \hat{\tau}^{\dagger}\hat{Z} \hat{\tau} \right ]^{\dagger} \right)^k  \right)
        =
        \label{stateTwirl}
        \dfrac{1}{2} \hat{\tau}^{\dagger} \sum_{\substack{k \in \{ 0,1 \}\\ \hat{\mathcal{K}}} } \left( \hat{Z}^k \hat{\tau} \hat{\mathcal{K}} \hat{\tau}^{\dagger} \vert 0 \rangle \langle 0 \vert \hat{\tau} \hat{\mathcal{K}}^{\dagger}\hat{\tau}^{\dagger}\hat{Z}^k \right) \hat{\tau},
    \end{align}
    where $\hat{\mathcal{K}}$ are the Kraus operators of $\mathcal{E}$.
    $\hat{\mathcal{K}}$ has a decomposition in the $\Gamma_{\hat{\tau}}$ basis (as it is a basis for $\mathbb{U}(4)$). Therefore, $\hat{\tau} \hat{\mathcal{K}} \hat{\tau}^{\dagger}$ has a decomposition in the Pauli basis with the same coefficients as $\hat{\mathcal{K}}$ has in the $\Gamma_{\hat{\tau}}$ basis. Using this decomposition in the RHS of Eqn.~\eqref{stateTwirl}, reduces the summation in the RHS of Eqn.~\eqref{stateTwirl} to the Pauli twirling of preparing the state $\vert 0 \rangle \langle 0 \vert$. Therefore, applying the results of Ref.~\cite{Ferracin_2019}:
    \begin{align}
        J
        &=
        \hat{\tau}^{\dagger} \sum_{k \in \{ 0,1 \} } \left( \alpha_k \hat{Z}^k \vert 0 \rangle \langle 0 \vert \hat{Z}^k \right) \hat{\tau}
        =
        \sum_{k \in \{ 0,1 \} } \left( \alpha_k  \left(\hat{\tau}^{\dagger} \hat{Z} \hat{\tau} \right)^{k} \vert \hat{\tau}^{\dagger} \rangle \langle \hat{\tau}^{\dagger} \vert \left( \hat{\tau}^{\dagger} \hat{Z} \hat{\tau} \right)^{k} \right)=
        \sum_{\hat{\gamma} \in \Gamma^*_{\hat{\tau}} } \left( \alpha_k \hat{\gamma} \vert \hat{\tau}^{\dagger} \rangle \langle \hat{\tau}^{\dagger} \vert \hat{\gamma}^{\dagger} \right),
    \end{align}
    where $\forall k \in \{0,1\}, \alpha_k \in [0,1]$. Let $\Gamma_{\hat{\tau}}^* = \big \{ \hat{\tau}^{\dagger} \hat{P} \hat{\tau} \text{ } \vert \text{ } \hat{P} \in \{ \hat{I}, \hat{X} \} \big \}$.
    Hence all error in the preparation of the state $\vert \hat{\tau}^{\dagger} \rangle \langle \hat{\tau}^{\dagger} \vert$ has been effectively reduced to stochastic $\Gamma^{*}_{\hat{\tau}} \subseteq \Gamma_{\hat{\tau}}$ error.
\end{proof}

\begin{lemma}
    \label{measTwirlError}
    For any $\hat{\tau} \in \mathbb{U}(2)$, CPTP error in a measurement in the $\{ \hat{\tau}^{\dagger} \vert 0 \rangle, \hat{\tau}^{\dagger} \vert 1 \rangle \}$ basis can be twirled to stochastic $\hat{\gamma}_{\hat{\tau}}$ error by, with probability $0.5$, applying a $\hat{\tau}^{\dagger} \hat{Z} \hat{\tau}$ gate immediately before the measurement.
\end{lemma}
\begin{proof}
    Define the probability of a single measurement measuring $\vert \hat{\tau}^{\dagger} \rangle$ as $\mathbb{P} \left( \vert \hat{\tau}^{\dagger} \rangle \right)$ and define $\rho_{\textit{nd}}$ as an arbitrary state the qubits are in immediately before the twirl is applied.
    Therefore, $\mathbb{P} \left( \vert \hat{\tau}^{\dagger} \rangle \right)$ can be evaluated (using the the Kraus decomposition of any CPTP error that may occur) as:
    \begin{align}
        \mathbb{P} \left( \vert \hat{\tau}^{\dagger} \rangle \right)
        &=
        \sum_{\substack{k \in \{0,1\}\\ \hat{\mathcal{K}}}} \left( \textit{Tr} \bigg[ \vert \hat{\tau}^{\dagger} \rangle \langle \hat{\tau}^{\dagger} \vert \hat{\mathcal{K}} \hat{\tau}^{\dagger} \hat{Z}^k \hat{\tau} \rho_{\textit{nd}} \hat{\tau}^{\dagger} \hat{Z}^k \hat{\tau} \hat{\mathcal{K}}^{\dagger} \bigg] \right)
        =
        \sum_{k \in \{0,1\}} \left( \textit{Tr} \bigg[ \sum_{\hat{\mathcal{K}}} \left(\sum_{\hat{P}_1, \hat{P}_2 \in \mathbb{P}_1} \left( \alpha^{(\hat{\mathcal{K}})}_{\hat{P}_1} \left( \alpha^{(\hat{\mathcal{K}})}_{\hat{P}_2} \right)^* \hat{\tau}^{\dagger} \vert 0 \rangle \langle 0 \vert \hat{P}_1 \hat{Z}^k \hat{\tau} \rho_{\textit{nd}} \hat{\tau}^{\dagger} \hat{Z}^k \hat{P}_2^{\dagger} \hat{\tau} \right) \bigg] \right) \right),
    \end{align}
     where the second equality follows from the fact that each $\hat{\mathcal{K}}$ can be decomposed into the $\Gamma^*_{\hat{\tau}}$ basis and -- similarly to before $\alpha_{\hat{P}_j}^{(\mathcal{K})} = \text{Tr}\left( \hat{P}_j \hat{\mathcal{K}}\right) / 2^{N}$. As $\rho_{\textit{nd}}$ is the result of a circuit that ends in $\hat{\tau}^{\dagger}$, define $\Bar{\rho}$ by $\rho_{\textit{nd}} = \hat{\tau}^{\dagger} \Bar{\rho} \hat{\tau}$
    \begin{align}
        \mathbb{P} \left( \vert \hat{\tau}^{\dagger} \rangle \right)
        &=
        \sum_{\substack{k \in \{0,1\}\\ \hat{\mathcal{K}}}} \left( \textit{Tr} \bigg[ \sum_{\hat{P}_1, \hat{P}_2 \in \mathbb{P}_1} \left( \alpha^{(\hat{\mathcal{K}})}_{\hat{P}_1} \left( \alpha^{(\hat{\mathcal{K}})}_{\hat{P}_2} \right)^* \vert 0 \rangle \langle 0 \vert \hat{P}_1 \hat{Z}^k \hat{\tau} \hat{\tau}^{\dagger} \Bar{\rho} \hat{\tau} \hat{\tau}^{\dagger} \hat{Z}^k \hat{P}_2^{\dagger} \right) \bigg] \right)
        =
        \sum_{\substack{k \in \{0,1\}\\ \hat{\mathcal{K}}}} \left( \textit{Tr} \bigg[ \Bar{\rho} \sum_{\hat{P}_1, \hat{P}_2 \in \mathbb{P}_1} \left( \alpha^{(\hat{\mathcal{K}})}_{\hat{P}_1} \left( \alpha^{(\hat{\mathcal{K}})}_{\hat{P}_2} \right)^* \hat{Z}^k \hat{P}_2 \hat{Z}^k \vert 0 \rangle \langle 0 \vert \hat{Z}^k \hat{P}_1 \hat{Z}^k \right) \bigg] \right).
    \end{align}
    Using the derivation of Pauli twirling state preparation in Ref.~\cite[Appendix A]{Ferracin_2019}, and defining $d_{\hat{P}_1} = \sum_{\hat{\mathcal{K}}} \left( \alpha^{(\hat{\mathcal{K}})}_{\hat{P}_1} \left( \alpha^{(\hat{\mathcal{K}})}_{\hat{P}_1} \right)^* \right) \in \mathbb{R}^+$,
    \begin{align}
        \mathbb{P} \left( \vert \hat{\tau}^{\dagger} \rangle \right)
        &=
        \textit{Tr} \bigg[ \Bar{\rho} \sum_{\hat{P}_1 \in \mathbb{P}_1} \left( d_{\hat{P}_1} \hat{P}_1 \vert 0 \rangle \langle 0 \vert \hat{P}_1 \right) \bigg]
        =
        \label{stochasError}
        \textit{Tr} \bigg[ \vert 0 \rangle \langle 0 \vert \sum_{\hat{\gamma}_1 \in \Gamma_{\hat{\tau}}} \left( d_{\hat{\gamma}_1} \hat{\gamma}_1 \rho_{\textit{nd}} \hat{\gamma}_1^{\dagger} \right) \bigg],
    \end{align}
    where $ \forall \hat{\gamma}_1 \in \Gamma_{\hat{\tau}}$ and corresponding $\hat{P}_1 \in \mathbb{P}_1^{\otimes 2}$, $d_{\hat{\gamma}_1} = d_{\hat{P}_1}$.
    Eqn.~\eqref{stochasError} represents measurement experiencing solely stochastic $\Gamma_{\hat{\tau}}$ error immediately preceeding it.
\end{proof}

\section{Lemmas Used in the Proof of Theorem~\ref{accreditationTheorem}}
\label{sec:theorem3usedLemmas}
\begin{lemma}
\label{simulable}
    All trap circuits generated via Algorithm \ref{trapBuildingAlg} output a classically and efficiently computable output with certainty if no error occurs (i.e., in the ideal case).
\end{lemma}
\begin{proof}
    Consider a single trap circuit, $\mathcal{C}_{\textit{tr}}$. All gates in the traps generated by Algorithm \ref{trapBuildingAlg} are either $\hat{\mathcal{G}}$, $\hat{\Delta}$, $\hat{\Delta}^{\dagger}$, $\hat{\tau}^{\dagger}_1 H \hat{\tau}_1$ gates or involved in twirling the $\hat{\mathcal{G}}$ gates.
    In the error-free case, the twirling has no effect on the output of the traps, so the gates involved in twirling can be disregarded and we act as if they are not there. Consider every $\hat{\mathcal{G}}$ gate as its $\hat{\tau}$-decomposition.
    
    Then $\mathcal{C}_{\textit{tr}}$ can be considered as consisting of Clifford gates, $\mathcal{M}$, as in the definition of $\hat{\tau}$-decomposible gates surrounded by $\hat{\tau}_1$, $\hat{\tau}_2$ gates, and their Hermitian conjugates. If the state prep and measurement are also decomposed into $\hat{Z}$ basis measurements and $\hat{\tau}_1$ gates, then, due to Algorithm~\ref{addDeltaAlg}'s  addition of the $\hat{\Delta}$, $\hat{\Delta}^{\dagger}$ gates (which can also be decomposed into $\hat{\tau}_j$ gates and their conjugates), every $\hat{\tau}_j$ gate is immediately neighbours with a $\hat{\tau}_j^{\dagger}$ gate, and vice versa. Due to this, all $\hat{\tau}_1$, $\hat{\tau}_2$, $\hat{\tau}^{\dagger}_1$, and $\hat{\tau}^{\dagger}_2$ gates cancel in $\mathcal{C}_{\textit{tr}}$, leaving just $\mathcal{M}$ gates acting on an initial state which it has no effect on (as assumed in the definition of $\hat{\tau}$-decompositions), then measured in the $Z$-basis. Hence, the trap circuit is equivalent to preparing a state, $\vert 0 \rangle^{\otimes N}$, and them immediately measuring in the basis the state was prepared in, the $Z$-basis. This gives a certain output.
\end{proof}

\begin{lemma}
    \label{lem:SameInIdealCase}
    Algorithm~\ref{targBuildingAlg} produces a circuit with exactly the same probabilities of each measurement outcome as in the input circuit, in the ideal case.
\end{lemma}
\begin{proof}
    Algorithm~\ref{targBuildingAlg} replaces each single-qubit state preparation with an alternate state but then adds a single-qubit gate that maps the new prepared state to the state prepared in the input circuit. Therefore, in the ideal case -- where no error acts -- the alternative state preparation and added single-qubit gate is indistinguishable from just preparing the state in the input circuit. A similar argument means that the changes to measurement made in  Algorithm~\ref{targBuildingAlg} to the input circuit have no effect, in the ideal case, either.

    Therefore, all that remains to show is that the gates added by Algorithm~\ref{targBuildingAlg} to twirl the error to stochastic error are effectively not there when there is no error. This follows from the gates used to twirl the error canceling with each other (in the case of the gates that twirl the error in other gates) or have no effect (in the case of the gates that twirl state preparation and measurement error).
\end{proof}

\begin{lemma}
    \label{toPauli}
    All error conforming to the error assumptions $\mathbf{N1}$ and $\mathbf{N2}$ in trap circuits can be considered as stochastic $\Gamma_{\hat{\tau}_1}$ error acting immediately before measurement.
\end{lemma}
\begin{proof}
    By Lemma \ref{twirlLemma} all error occurring in a $\hat{\mathcal{G}}$ gate can be considered as stochastic $\Gamma_{\hat{\tau}_1}$ or $\Gamma_{\hat{\tau}_2}$ error. Similarly, by Lemmas \ref{statePrepTwirlLemma} and \ref{measTwirlError} (in Appendix \ref{TwirlingTauDeets}), all error in state preparation and measurement can be considered as stochastic $\Gamma_{\hat{\tau}_1}$ or $\Gamma_{\hat{\tau}_2}$ error.\\
    Via the application of the twirls showcased in Fig.~\ref{surroundError}, in trap circuits, the (now stochastic) error can be ``pushed\footnote{Meaning that the circuit is rewritten as an equivalent circuit (i.e., giving the exact same measurement probabilities in the ideal case) but with all the error appearing at the end of the circuit, immediately before the measurement.}"~\cite{https://doi.org/10.48550/arxiv.2206.00215} to the end of the trap circuit, to immediately before the measurements whilst remaining elements of $\Gamma_{\hat{\mathcal{G}}}$ -- as all two-qubit gates in the circuit normalize the gates in $\Gamma_{\hat{\tau}_1}$ or $\Gamma_{\hat{\tau}_2}$ incident on them. The remaining stochastic error will be $\Gamma_{\hat{\tau}_1}$ error due to the last step in Algorithm~\ref{addDeltaAlg}.
\end{proof}
\begin{lemma}
\label{allDetectLemma}
    If the $\hat{\tau}^{\dagger}_1 H \hat{\tau}_1$ gates are applied in Algorithm~\ref{trapBuildingAlg},
    all stochastic errors that could not be detected by the traps without the $\hat{\tau}^{\dagger}_1 H \hat{\tau}_1$ gates are mapped to errors that can be detected by the traps. 
\end{lemma}
\begin{proof}
    Due to Lemma~\ref{toPauli}, all error can then be considered as stochastic $\hat{\tau}^{\dagger}_1 \hat{X} \hat{\tau}_1$ error and $\hat{\tau}^{\dagger}_1 \hat{Z} \hat{\tau}_1$ error. The only one of these that does not change the outcome of a measurement in the $\vert \hat{\tau}_1 \rangle$ basis is $\hat{\tau}^{\dagger}_1 \hat{Z} \hat{\tau}_1$. All error in the state preparation and measurement is twirled to exclusively stochastic $\hat{\tau}^{\dagger}_1 \hat{X} \hat{\tau}_1$ error, hence all $\hat{\tau}^{\dagger}_1 \hat{Z} \hat{\tau}_1$ error in a circuit occurs in the gates. As per Algorithm~\ref{trapBuildingAlg}, all gates in the circuit are contained within a two layers of $\hat{\tau}^{\dagger}_1 H \hat{\tau}_1$ on each qubit. $\hat{\tau}^{\dagger}_1 \hat{Z} \hat{\tau}_1$ error is converted to $\hat{\tau}^{\dagger}_1 \hat{X} \hat{\tau}_1$ error as it is pushed through the $\hat{\tau}^{\dagger}_1 H \hat{\tau}_1$ gates as:
    \begin{align}
        \left( \hat{\tau}^{\dagger}_1 H \hat{\tau}_1 \right) \left( \hat{\tau}^{\dagger}_1 \hat{Z} \hat{\tau}_1 \right)
        &=
        \hat{\tau}^{\dagger}_1 H \hat{Z} \hat{\tau}_1
        =
        \hat{\tau}^{\dagger}_1 \hat{X} H \hat{\tau}_1
        =
        \left(\hat{\tau}^{\dagger}_1  \hat{X} \hat{\tau}_1 \right) \left( \hat{\tau}^{\dagger}_1 H \hat{\tau}_1 \right).
    \end{align}
\end{proof}
\begin{Note}
    Due to Lemma~\ref{toPauli}, given assumptions N1 and N2, all error can be considered to be stochastic Pauli error. So Lemma~\ref{allDetectLemma} applies to all error occurring in trap and target circuits.
\end{Note}

\section{Proof of Theorem~\ref{TwoFoldTwirlTheorem}}
\label{LargerTwirlAppendix}
\label{ProofXyTwirlTheoremAppendix}

Before proving Theorem~\ref{TwoFoldTwirlTheorem}, we first require several lemmas and a corollary that will be used in the proof of Theorem~\ref{TwoFoldTwirlTheorem}. These begin with Corollary~\ref{twirlconditionsLemma}.
\begin{corollary}
    \label{twirlconditionsLemma}
    $\forall N \in \mathbb{N}$, $\forall \Lambda, Q \in \mathbb{P}_1^{\otimes N}$, conjugation by uniformly random elements from $\Lambda$ effectively reduces any CPTP error, where all terms with non-zero coefficients in the error's Kraus operators Pauli decompositions are in Q, to stochastic Pauli error, if: $\forall \hat{\gamma}_1, \hat{\gamma}_2 \in Q$,
    \begin{align}
        \sum_{\hat{\lambda} \in \Lambda} \left( \xi \left( \hat{\gamma}_1 \circ \hat{\gamma}_2, \hat{\lambda} \right) \right)
        &=
        \delta_{\hat{\gamma}_1, \hat{\gamma}_2} \big \vert \Lambda \big \vert.
    \end{align}
\end{corollary}
\begin{proof}
    The core of this proof is an application of Lemma~\ref{MostGeneralTwirlLemma}. Therefore, this proof, in essence, consists of showing that the conditions of Lemma~\ref{MostGeneralTwirlLemma} are met.
    
    Conditions $\mathbf{S1}$ and $\mathbf{S2}$ of Lemma \ref{MostGeneralTwirlLemma} are met as $\forall N \in \mathbb{N}$, $\mathbb{P}_1^{\otimes N}$ is closed under composition and any two elements in $\mathbb{P}_1^{\otimes N}$ either commute or anti-commute.

    Similarly, the $\Lambda$ self-adjoint condition is met as $\Lambda \subseteq \mathbb{P}_1^{\otimes N}$ and all elements of $\mathbb{P}_1^{\otimes N}$ are self-adjoint.
    
    Finally, if $\forall \hat{\gamma}_1, \hat{\gamma}_2 \in Q$,
    \begin{align}
        \sum_{\hat{\lambda} \in \Lambda} \left( \xi \left( \hat{\gamma}_1 \circ \hat{\gamma}_2, \hat{\lambda} \right) \right)
        &=
        \delta_{\hat{\gamma}_1, \hat{\gamma}_2} \big \vert \Lambda \big \vert,
    \end{align}
    then Lemma \ref{MostGeneralTwirlLemma} implies the CPTP 
error conjugated by uniformly random elements of $\Lambda$ is reduced to stochastic Pauli error.
\end{proof}

We now begin the lemmas of this appendix. The first of these is Lemma~\ref{identLemma}, which will be then used in Lemma~\ref{compositionLemma}.

\begin{lemma}
    \label{identLemma}
    For any $\hat{\gamma}_1, \hat{\gamma}_2 \in \mathbb{K}$, where $\mathbb{K}$ is an unflipable set as in the definition of XY-twirlable error (Def.~\ref{XYTwirlableDef}),
    \begin{align}
        \hat{\gamma}_1 \circ \hat{\gamma}_2 = \hat{I} \iff \hat{\gamma}_1 = \hat{\gamma}_2.
    \end{align}
\end{lemma}
\begin{proof}\( \\ \)
    \underline{$\Leftarrow$}:
    $\mathbb{K} \subset \mathbb{P}_1^{\otimes 2} \subset \mathbb{U}(4)$, therefore, $\forall \hat{\gamma}_1 \in \mathbb{K}$, $\hat{\gamma}_1 = \hat{\gamma}_1^{-1}$. Hence, if $\hat{\gamma}_1 = \hat{\gamma}_2$, then $
        \hat{\gamma}_1 \circ \hat{\gamma}_2
        =
        \hat{\gamma}_1^{-1} \circ \hat{\gamma}_1
        =
        \hat{I}
        $.\\
    \underline{$\Rightarrow$}:
    Using the same properties of $\mathbb{K}$ as above:
    $
        \hat{\gamma}_1 \circ \hat{\gamma}_2 = \hat{I}
        \Rightarrow
        \hat{\gamma}_1^{-1} \circ \hat{\gamma}_1 \circ \hat{\gamma}_2 = \hat{\gamma}_1^{-1} \circ \hat{I}
        \Rightarrow
        \hat{\gamma}_2 = \hat{\gamma}_1^{-1} = \hat{\gamma}_1
    $.
\end{proof}

Before proceeding to Lemma~\ref{compositionLemma}, we first define notation for a -- unwieldily large -- set that we will use in Lemma~\ref{compositionLemma}.
\begin{definition}
    Define the set:
        $\underline{\mathbb{D}}
        =
        \big \{\hat{X} \otimes \hat{Y}, \hat{X} \otimes \hat{Z}, \hat{Y} \otimes \hat{X}, \hat{Y} \otimes \hat{Z}, \hat{Z} \otimes \hat{X}, \hat{Z} \otimes \hat{Y}, \hat{I} \otimes \hat{X}, \hat{I} \otimes \hat{Y}, \hat{I} \otimes \hat{Z}, \hat{X} \otimes \hat{I}, \hat{Y} \otimes \hat{I}, \hat{Z} \otimes \hat{I}, \hat{I} \otimes \hat{I} \big \} \subset \mathbb{P}_1^{\otimes 2}$.
\end{definition}

\begin{lemma}
    \label{compositionLemma}
    For any unflippable set, $\mathbb{K}$, $\forall \hat{a}, \hat{b} \in \mathbb{K}$, 
    \begin{align}
        \hat{a} \circ \hat{b} \in \mathbb{D}.
    \end{align}
\end{lemma}
\begin{proof}
    Let $\hat{a}, \hat{b} \in \mathbb{K}$, then there exists $j_a, k_a, j_b, k_a \in \{ 0, 1, 2, 3\}$, such that:
    \begin{align}
        \label{eqn:reexpressAoB}
         a = \hat{\sigma}_{j_a} \otimes \hat{\sigma}_{k_a} &\text{ and } b = \hat{\sigma}_{j_b} \otimes \hat{\sigma}_{k_b},\\
         &\text{where} \nonumber\\
         j_a \not = k_a &\text{ and } j_b \not = k_b. \label{jBkBInequality}
    \end{align}
    Eqn.~\eqref{eqn:reexpressAoB} and  Eqn.~\eqref{jBkBInequality} can then be used to deduce:
    \begin{align}
        \label{eqn:AcircBMultipliedOut}
        a \circ b
        =
        \left( \hat{\sigma}_{j_a} \otimes \hat{\sigma}_{k_a} \right) \circ \left( \hat{\sigma}_{j_b} \otimes \hat{\sigma}_{k_b} \right)
        =
        \left( \hat{\sigma}_{j_a} \circ \hat{\sigma}_{j_b} \right) \otimes \left( \hat{\sigma}_{k_a} \circ \hat{\sigma}_{k_b} \right)
        =
        \left( \delta_{j_a, j_b} I + i \mathcal{E}_{j_a, j_b, l} \hat{\sigma}_l \right) \otimes \left( \delta_{k_a, k_b} I + i \mathcal{E}_{k_a, k_b, q} \hat{\sigma}_q \right).
    \end{align}
    Assume, for the sake of contradiction, that Lemma~\ref{compositionLemma} is false, i.e., $\hat{a} \circ \hat{b} \not \in  \mathbb{D}$. This is equivalent to: $\hat{a} \circ \hat{b} \in \big \{ \hat{X} \otimes \hat{X}, \hat{Y} \otimes \hat{Y}, \hat{Z} \otimes \hat{Z} \big \}$.
    
    This assumption and Eqn.~\eqref{eqn:AcircBMultipliedOut} imply:
        \begin{align}
        \label{eqn:derivingSetofOptions}
        \mathcal{E}_{j_a, j_b, l} \hat{\sigma}_l
        &=
        \mathcal{E}_{k_a, k_b, q} \hat{\sigma}_q
        \Rightarrow 
        \big \{ j_a, j_b \big \} = \big \{ k_a, k_b \big \}.
        \end{align}
         Eqn.~\eqref{jBkBInequality} and Eqn.~\eqref{eqn:derivingSetofOptions} can only be consistent with each other if $j_a = k_b \hspace{0.1cm} \textit{and} \hspace{0.1cm} j_b = k_a$. Returning to Eqn.~\eqref{eqn:reexpressAoB} with this newfound equality:
        \begin{align}
        \hat{a} = \hat{\sigma}_{j_a} \otimes \hat{\sigma}_{j_b} \hspace{0.1cm} \textit{and} \hspace{0.1cm} \hat{b} = \hat{\sigma}_{j_b} \otimes \hat{\sigma}_{j_a}.
        \end{align}
    As both $\hat{a}$ and $\hat{b}$ are defined to be in the unflippable set, $\mathbb{K}$, this is a contradiction as it implies $\mathbb{K}$ is not unflippable (which it is defined to be).
    Hence, $\hat{a} \circ \hat{b} \in \mathbb{D}$.
\end{proof}

We have finally reached the point where we are ready to present the proof of Theorem~\ref{TwoFoldTwirlTheorem}. For convenience, we restate Theorem~\ref{TwoFoldTwirlTheorem} before the proof.
\begin{theoremRestate}[Theorem~\ref{TwoFoldTwirlTheorem}]
    All XY-twirlable error occurring in a XY-decomposable gate can be twirled to stochastic Pauli error by surrounding it by Pauli gates.
\end{theoremRestate}
\begin{proof}[Proof of Theorem~\ref{TwoFoldTwirlTheorem}]
    Corollary \ref{twirlconditionsLemma} implies Theorem~\ref{TwoFoldTwirlTheorem} is true if for any unflippable set, $Q \subseteq \mathbb{P}_1^{\otimes 2}$ (as defined in Def.~\ref{XYTwirlableDef}), there exists a $\Lambda \subseteq \mathbb{P}_1^{\otimes 2}$ such that:
    $\forall \hat{\gamma}_1, \hat{\gamma}_2 \in Q$,
    \begin{align}
        \label{eqn:FinalLambdaConditionUse}
        \sum_{\hat{\lambda} \in \Lambda} \left( \xi \left( \hat{\gamma}_1 \circ \hat{\gamma}_2, \hat{\lambda} \right) \right)
        &=
        \delta_{\hat{\gamma}_1, \hat{\gamma}_2} \big \vert \Lambda \big \vert.
    \end{align}
    
    We now show that, for any unflippable set, $\Lambda = \{ \hat{I} \otimes \hat{I}, \hat{X} \otimes \hat{X}, \hat{Y} \otimes \hat{Y}, \hat{Z} \otimes \hat{Z} \}$ meets the required conditions (to be a valid $\Lambda$ in Eqn.~\eqref{eqn:FinalLambdaConditionUse}). We split proving Eqn.~\eqref{eqn:FinalLambdaConditionUse} holds into two cases:
    \begin{enumerate}
        \item $\hat{\gamma}_1 \not = \hat{\gamma}_2$
        \item $\hat{\gamma}_1 = \hat{\gamma}_2$
    \end{enumerate}
    In the first case, Lemma \ref{compositionLemma} shows: $\forall \hat{\gamma}_1, \hat{\gamma}_2 \in Q \subseteq \mathbb{P}^{\otimes 2}_1$,
    \begin{align} 
    \hat{\gamma}_1 \circ \hat{\gamma}_2 \in \big \{ \hat{X} \otimes \hat{Y}, \hat{X} \otimes \hat{Z}, \hat{Y} \otimes \hat{X}, \hat{Y} \otimes \hat{Z}, \hat{Z} \otimes \hat{X}, \hat{Z} \otimes \hat{Y}, \hat{I} \otimes \hat{X}, \hat{I} \otimes \hat{Y}, \hat{I} \otimes \hat{Z}, \hat{X} \otimes \hat{I}, \hat{Y} \otimes \hat{I}, \hat{Z} \otimes \hat{I} \big \}.
    \end{align}
    
    The summation in Eqn.~\eqref{eqn:FinalLambdaConditionUse} can be shown to equate to zero in this case by summing each column in Table~\ref{signsTable} and checking they \emph{all} sum to zero \footnote{To make Table~\ref{signsTable} easier to read, as all entries are either +1 or -1, we just show the sign of the resulting $\xi$. Therefore, the table can be easily checked by confirming every column in Table~\ref{signsTable} contains exactly two positive signs and exactly two minus signs}.
    \begin{table}
\centering
%\newcolumntype{c}{>{\centering\arraybackslash}X}
\begin{tabular}{c|cccccccccccc}
$\xi$         & $\hat{X} \otimes \hat{Y}$ & $\hat{X} \otimes \hat{Z}$ & $\hat{Y} \otimes \hat{X}$ & $\hat{Y} \otimes \hat{Z}$ & $\hat{Z} \otimes \hat{X}$ & $\hat{Z} \otimes \hat{Y}$ & $\hat{I} \otimes \hat{X}$ & $\hat{I} \otimes \hat{Y}$ & $\hat{I} \otimes \hat{Z}$ & $\hat{X} \otimes \hat{I}$ & $\hat{Y} \otimes \hat{I}$ & $\hat{Z} \otimes \hat{I}$\\ 
\hline
$\hat{I} \otimes \hat{I}$ & +                            & +                            & +                           & +                              & +                             & +  
& + & + & +
& + & + & +\\
$\hat{X} \otimes \hat{X}$ & -                             & -                              & -                             &+                             & -                              & +
& + & - & - 
& + & - & -\\
$\hat{Y} \otimes \hat{Y}$ & -                              & +                             & -                              & -                             & +                             & - 
& - & + & -
& - & + & - \\
$\hat{Z} \otimes \hat{Z}$ & +                              & -                              & +                              & -                              & -                              & - & - & - & +
& - & - & +
\end{tabular}\caption{Table showing the signs of $\xi \left( \hat{A}, \hat{B} \right)$, where $\hat{A} = \hat{\gamma}_1 \circ \hat{\gamma}_2$ for some $\hat{\gamma}_1, \hat{\gamma}_2 \in Q$ (where $\hat{\gamma}_1 \not = \hat{\gamma}_2$) and $\hat{B} \in \Lambda$.\\
The elements of $\Lambda$ are along the left-most column and the elements of $\mathbb{D}$ (as Lemma~\ref{compositionLemma} shows $\hat{A} = \hat{\gamma}_1 \circ \hat{\gamma}_2 \in \mathbb{D}$) are along the top row. By checking that all columns sum to zero (i.e., have an equal number of positive and negative signs) we can confirm that, regardless of the value of $\hat{A} = \hat{\gamma}_1 \circ \hat{\gamma}_2$, if $\hat{\gamma}_1 \not = \hat{\gamma}_2$ Eqn~\ref{eqn:FinalLambdaConditionUse} is satisfied.}
\label{signsTable}
\end{table}

All that remains is the second case: that the summation is $\big \vert \Lambda \big \vert$ when $\hat{\gamma}_1 = \hat{\gamma}_2$. Lemma~\ref{identLemma} shows in this case $\hat{\gamma}_1 \circ \hat{\gamma}_2 = \hat{I}$, hence $\forall \hat{\lambda} \in \Lambda$, $\xi \left( \hat{\gamma}_1 \circ \hat{\gamma}_2, \hat{\lambda} \right) = \xi \left( \hat{I}, \hat{\lambda} \right) = 1$. So the summation in Eqn.~\eqref{eqn:FinalLambdaConditionUse} equates to $\big \vert \Lambda \big \vert$, as required.
\end{proof}

\section{Lemmas Used in the Proof of Theorem~\ref{XYAccCentralTheorem}}
\label{sec:ProofOfTheo5Lemmas}
\begin{lemma}
\label{noErrorTrap}
    An error-free implementation of a trap circuit returns all $\vert 0 \rangle$ with probability one.
\end{lemma}
\begin{proof}
    Let $\Tilde{\nu}_k$ denote the $k$th band of vanishing blocks in the target circuit and assume the circuit has $N$ qubits. Then the probability of measuring $\vert 0 \rangle^{\otimes N}$ is:
    \begin{align}
        \mathbb{P} \left( \vert 0 \rangle^{\otimes N} \right)
        &=
        \big \vert \langle 0 \vert^{\otimes N} \Bar{Z} H^h \prod_{k = B}^1 \left( \Tilde{\nu}_k \right) H^h \Bar{Z}\vert 0 \rangle^{\otimes N} \big \vert^2,
    \end{align}
    where $\Bar{Z}$ denotes a Pauli randomly and independently -- with probability $0.5$ -- on each qubit, $h$ is uniformly from $\{ 0, 1\}$, and $B$ is the number of bands in the circuit.
    But as all vanishing blocks in a trap are $1$-vanishing blocks, $\forall k \in \mathbb{N}^{\leq B}$, $\Tilde{\nu}_k = \hat{I}$. Therefore,
    \begin{align}
        \mathbb{P} \left( \vert 0 \rangle^{\otimes N} \right)
        &=
        \big \vert \langle 0 \vert^{\otimes N} \Bar{Z} H^h H^h \Bar{Z}\vert 0 \rangle^{\otimes N} \big \vert^2
        =
        1.
    \end{align}
\end{proof}
With the measurement outcomes in the error-free cases of traps established, we can consider what happens in a trap when error is present. 

\begin{lemma}
    \label{XYStochasticInAcc}
    Assuming $\mathbf{N1}$ and $\mathbf{N2}$, all error in traps generated by Algorithm~\ref{TrapGenAlgorithm} can be considered as stochastic Pauli error occurring outside of vanishing blocks.
\end{lemma}
\begin{proof}
    Due to Lemma~\ref{StochErrorInVanishing}, all error not occurring in state preparation or measurement can be considered as stochastic Pauli error occurring outside of vanishing blocks.
    
    Lemmas~\ref{statePrepTwirlLemma} and~\ref{measTwirlError} (in Appendix~\ref{TwirlingTauDeets}) show (by setting $\hat{\tau} = \hat{I}$) that the Pauli $\hat{Z}$ gates added in steps 5 and 6 of Algorithm~\ref{TrapGenAlgorithm} twirl all error occurring in state preparation or measurement to stochastic Pauli $\hat{X}$ error (which is outside the vanishing blocks).
\end{proof}

\begin{lemma}
    \label{XYStochasticDetectLemma}
    All stochastic Pauli error occurring outside the vanishing blocks in a trap generated by Algorithm~\ref{TrapGenAlgorithm} is detected by that traps with probability at least $0.5$.
\end{lemma}
\begin{proof}
    Traps generated by Algorithm~\ref{TrapGenAlgorithm} feature only $1$-vanishing blocks and Hadamard gates (if they are present).
    
    As $1$-vanishing blocks are equivalent to the identity (as shown in Lemma~\ref{IdealJOneLemma}), a trap generated by Algorithm~\ref{TrapGenAlgorithm} experiencing stochastic Pauli error occurring outside a vanishing block is equivalent to just the stochastic Pauli error and possibly the Hadamard gates added in step 4 of Algorithm~\ref{TrapGenAlgorithm}. 
    
    Hence, neglecting error cancelling and using Lemma~\ref{allDetectLemma} (by setting $\hat{\tau}_1 = \hat{I}$), all stochastic Pauli error in the trap can be moved to the end of the trap and pushed through the last layer of Hadamard gates (if they are present, which happens with probability $0.5$). Hence Pauli $\hat{X}$ error reaches the measurements and flips the outcome with probability $0.5$ and Pauli $\hat{Z}$ error is transformed to Pauli $\hat{X}$ before it reaches the measurements and flips the outcome with probability $0.5$ (Pauli $\hat{Y}$ error always flips the outcome regardless).
    
    Therefore, \emph{any} Pauli error occurring in a trap is detected by that trap with probability at least $0.5$. 
\end{proof}
\end{document}